%% file: main.tex
\documentclass{article}
\input{preamble}
\usepackage[margin=1in]{geometry}
\usepackage[font=small,labelfont=sc]{caption}

\newcommand\extrafootertext[1]{%
    \bgroup
    \renewcommand\thefootnote{\fnsymbol{footnote}}%
    \renewcommand\thempfootnote{\fnsymbol{mpfootnote}}%
    \footnotetext[0]{#1}%
    \egroup
}

\title{Learnable Mixed Nash Equilibria are Collectively Rational}

\author{Geelon So \and Yi-An Ma}
\date{\large University of California, San Diego}

\usepackage[most]{tcolorbox}
\newtcolorbox{graybox}{
  colback=gray!5,   % background color (10% gray)
  colframe=black!30, % border color
  boxrule=0.5pt,     % border thickness
  arc=3pt,           % rounded corners
  left=6pt, right=6pt, top=6pt, bottom=6pt % padding
}

\providecommand{\strats}{\ensuremath{\mathbf{x}}}       % Joint strategy
\providecommand{\NE}{\ensuremath{\mathbf{x}^*}}        % Nash equilibrium
\providecommand{\SE}{\ensuremath{\mathbf{x}^\beta}}    % Smooth equilibrium

\providecommand{\domains}{\ensuremath{\mathbf{\Omega}}}
\providecommand{\utilities}{\ensuremath{\mathbf{f}}}
\providecommand{\BR}{\ensuremath{\boldsymbol{\Phi}}}
\providecommand{\betaBRs}{\ensuremath{\mathbf{\Phi}^\beta}}
\providecommand{\betaBR}{\ensuremath{\Phi^\beta}}
\providecommand{\regs}{\ensuremath{\mathbf{h}}}
\providecommand{\Jacobian}{\ensuremath{\mathbf{J}}}

\providecommand{\Hs}{\ensuremath{\mathbf{H}}}

\providecommand{\Lambdas}{\boldsymbol{\Lambda}}

\begin{document}

\maketitle

\begin{abstract}
    We extend the study of learning in games to dynamics that exhibit \emph{non-asymptotic stability}.  We do so through the notion of \emph{uniform stability}, which is concerned with equilibria of individually utility-seeking dynamics. Perhaps surprisingly, it turns out to be closely connected to economic properties of collective rationality. 
    Up to strategic equivalence, if a mixed equilibrium is uniformly stable, then it is \emph{weakly Pareto optimal}---there is no way for all players to improve by jointly deviating from the equilibrium---a form of collective rationality that rules out the types of behaviors in the \emph{prisoner's dilemma} or the \emph{tragedy of the commons}.
    Moreover, we show that uniform stability determines the last-iterate convergence behavior for the family of \emph{incremental smoothed best-response dynamics}, used to model individual and corporate behaviors in the markets. Unlike dynamics around strict equilibria, which can stabilize to socially-inefficient solutions, individually utility-seeking behaviors near mixed Nash equilibria lead to collective rationality.

   \vspace{3pt}\noindent{\footnotesize\textbf{Keywords:} algorithmic game theory, evolutionary dynamics, non-asymptotic stability, last-iterate convergence, smoothed best-response}
   \extrafootertext{Correspondence to: \texttt{geelon@ucsd.edu} and \texttt{yianma@ucsd.edu}}
\end{abstract}

\section{Introduction}

\begin{figure}
    \centering
    \includegraphics[width=0.85\linewidth]{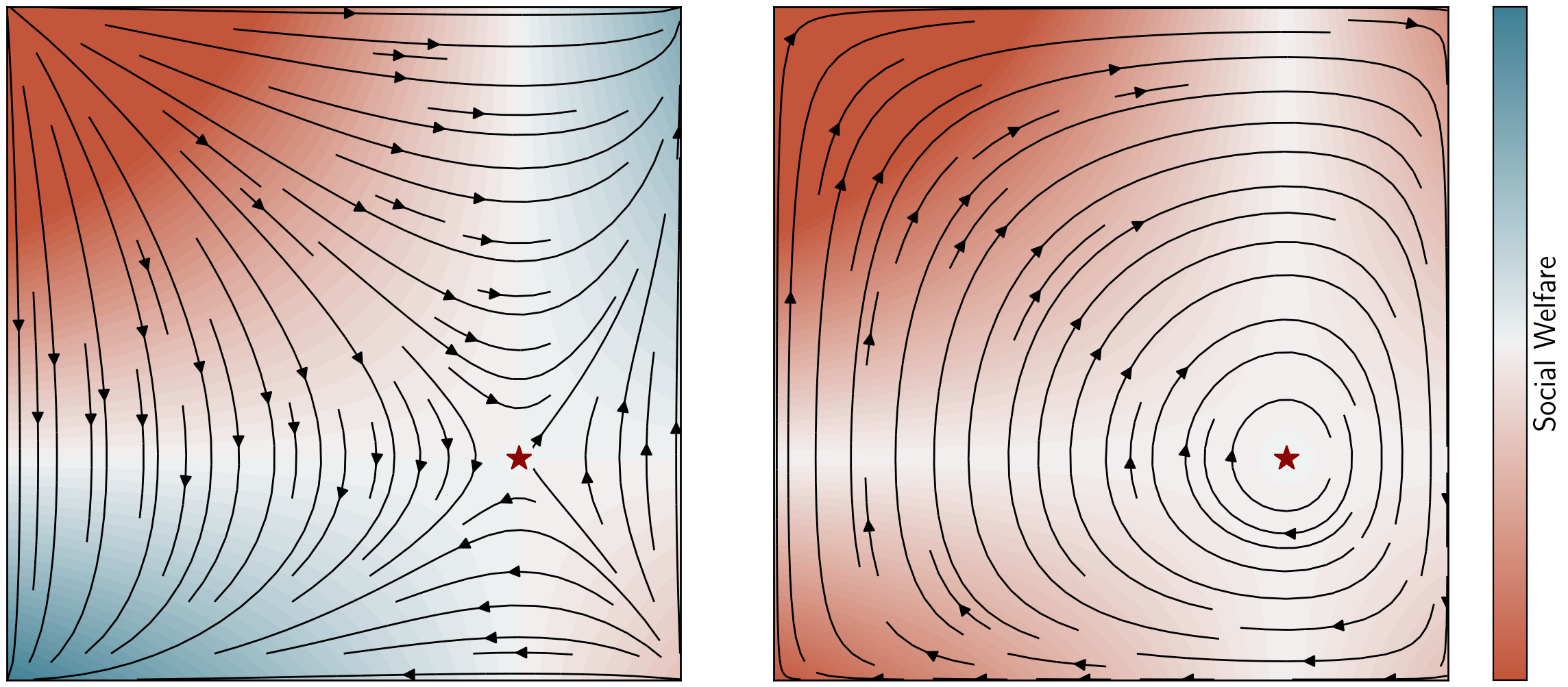}    
    \begin{minipage}{0.98\textwidth}
        \small \hspace{30pt}(a) Pareto-dominated equilibria are unstable
        \hspace{11pt}(b) Weakly-Pareto equilibria are neutrally stable
    \end{minipage}
    \caption{\small Uncoupled learning dynamics in two $2 \times 2$ normal-form games with purely-strategic utilities $\utilities = (f_1, f_2)$. The Nash equilibria are marked by stars. The streamlines visualize the trajectories of the learning dynamics. The heatmap plots the social welfare function $\min\{f_1, f_2\}$, measuring the utility of the player worst-off. The heatmap is white where the utilities are equal to the equilibrium; the darker the red, the worse the social welfare; the darker the blue, the better. (a) The mixed Nash equilibrium in this game is not weakly Pareto optimal, and the dynamics are unstable. (b) This equilibrium is weakly Pareto optimal, and the dynamics are non-asymptotically stable. The dynamics visualized here is continuous-time mirror ascent induced by the entropy mirror map.}
    \label{fig:two-behaviors}
\end{figure}

\input{doc/1-intro}

\section{Preliminaries}
\input{doc/2-preliminaries}

\section{Uniform stability of Nash equilibria} \label{subsec:equilibrium-classes}
\input{doc/3-uniform-stability}

\section{The economic meaning of uniform stability}
\label{sec:economic-meaning}
\input{doc/4-economic-meaning}

\section{Convergence under uniform stability}
\label{sec:convergence}
\input{doc/5-convergence-mixed}

\section{Uniform stability for partially-mixed equilibria}
\label{sec:boundary}
\input{doc/6-convergence-partial}

\section{Discussion}
\input{doc/7-discussion}

\section*{Acknowledgments}
We express our gratitude toward Mike Jordan, for his inspiration and suggestions.
\newpage

\bibliography{references}

\newpage

\renewcommand{\theHlemma}{\thesection.\arabic{lemma}} 
\renewcommand{\theHtheorem}{\thesection.\arabic{theorem}} 
\renewcommand{\theHdefinition}{\thesection.\arabic{definition}} 
\renewcommand{\theHcorollary}{\thesection.\arabic{corollary}} 
\appendix

\section{Normal-form games and multilinear games} \label{sec:normal-form-games-exposition}

\subsection{Background for normal-form games}
In a \emph{normal-form game}, each player $n \in [N]$ chooses an action out of $k_n$ alternatives or pure strategies, indexed over $[k_n]$. To allow for randomized strategies, let $\Delta^{k_n-1}$ denote the \emph{probability simplex} over $[k_n]$,
\[\Delta^{k_n - 1} = \Big\{x_n \in \RR^{k_n} : \sum_{i \in [k_n]} x_{n,i} = 1 \textrm{ and } x_{n,i} \geq 0\Big\}.\]
The $n$th player's strategy space is $\Omega_n = \Delta^{k_n -1}$, where we say that $x_n \in \Omega_n$ is a \emph{mixed strategy profile}, with: 
\[\forall i \in [k_n],\qquad x_{n,i} = \textrm{probability that player $n$ plays action $i$}.\hspace{1em}\]
Assuming players randomize their actions independently, each joint mixed strategy $\strats = (x_1,\ldots, x_N) \in \domains$ defines the product distribution $x_1 \otimes \dotsm \otimes x_N$ over the joint space of alternatives $[k_1] \times \dotsm \times [k_N]$.

The utilities in a normal-form game are defined by a collection of \emph{payoff tensors} $T^1,\ldots, T^N \in \RR^{k_1 \times \dotsm \times k_N}$,
\[T^n_{i_1,\ldots,i_N} = \textrm{payoff for player $n$ when the $m$th player chooses action $i_m$ for $m \in [N]$}.\]
The utilities $\utilities$ of a normal-form game extend these payoffs to mixed strategies in all of $\domains$, and they compute the \emph{expected payoffs} when each player randomize independently:
\[f_n(\strats) = \sum_{i_1 \in [k_1]}\dotsm\sum_{i_N \in [k_N]} T^n_{i_1,\ldots, i_N} x_{1,i_1} \dotsm x_{N,i_N}.\]

\begin{lemma}[Normal-form games are multilinear games] \label{lem:normal-to-multilinear}
    Let $(\domains, \utilities)$ be an $N$-player normal-form game, where the $n$th player has $k_n$ alternatives. There is a multilinear game $(\tilde{\domains}, \tilde{\utilities})$ with $\tilde{\utilities} : \RR^{k_1 -1} \times \dotsm \times \RR^{k_N -1} \to \RR^N$ and an affine bijection $\phi : \mathrm{int}(\domains) \to \tilde{\domains}$ such that $\utilities(\strats) = (\tilde{\utilities} \circ \phi)(\strats)$ for all $\strats \in \mathrm{int}(\domains)$.
\end{lemma}

\begin{proof}[Proof of \Cref{lem:normal-to-multilinear}]
    This follows from the form of the utility of a normal form game.
\end{proof}

\subsection{Individual and collective rationality are generally incomparable}

\begin{example}[Strategic Pareto optimality does not imply Nash stability] \label{ex:strategic-not-nash}
    Consider the following 2-player normal-form game, which is purely strategic, since all rows and columns sum to zero:
    \begin{center}
    \begin{tabular}{c||c|c|c}
        $f_A$ & $B_1$ & $B_2$ & $B_3$ \\
        \hline\hline
        $A_1$ & $4$ & $-1$\, & $-3$\, \\
        \hline
        $A_2$ & $-10$\, & $\mathbf{2}$ & $\mathbf{8}$ \\
        \hline
        $A_3$ & $\mathbf{6}$ & $-1$\, & $-5$\, \\
    \end{tabular}
    \qquad
    \begin{tabular}{c||c|c|c}
        $f_B$ & $B_1$ & $B_2$ & $B_3$ \\
        \hline\hline
        $A_1$ & $4$ & $-10$\, & $\mathbf{6}$ \\
        \hline
        $A_2$ & $-1$\, & $\mathbf{2}$ & $-1$\, \\
        \hline
        $A_3$ & $-3$\, & $\mathbf{8}$ & $-5$\, \\
    \end{tabular}
    \end{center}
    Here, we show the payoff matrices for the $A$ and $B$ players, and bold the preferred action conditioned on the other player. In particular, the joint strategy $(A_2, B_2)$ is a Nash equilibrium, but not a strategic Pareto optimum, since it can be improved to $(A_1, B_1)$. And while $(A_1, B_1)$ is Pareto optimal, it is not a Nash equilibrium.
\end{example}

\section{Smoothed equilibria and dynamics}

\begin{lemma}[$\beta$-smoothed equilibria are $O(\beta)$-approximate Nash equilibria] \label{lem:smooth-to-Nash}
    Let $\regs$ be a set of strictly convex and bounded regularizers, where $\max h_n - \min h_n \leq C$ for $n \in [N]$, inducing the $\beta$-smoothed best-response map $\betaBRs$. If a joint strategy is a $\beta$-smoothed equilibrium of $\betaBRs$, then it is a $C\beta$-approximate Nash equilibrium.
\end{lemma}

%\subsection[Proof of Lemma \ref*{lem:smooth-to-Nash}]{Proof of \Cref{lem:smooth-to-Nash}}

\begin{proof}[Proof of \Cref{lem:smooth-to-Nash}]
    Let $\strats$ be a $\beta$-smoothed equilibrium with respect to $\regs$. By the definition of $\betaBRs$, we have that the following inequality holds for any alternate strategy $z_n \in \Omega_n$;
    \[f_n(z_n; \strats_{-n}) - \beta h_n(z_n) \leq f_n(\strats) - \beta h_n(x_n).\]
    By rearranging and applying the boundedness of the regularizer, we obtain:
    \begin{align*} 
        f_n(z_n; \strats_{-n}) &\leq f_n(\strats) + \beta \big(h_n(z_n) - h_n(x_n)\big) 
        \\&\leq f_n(\strats) + C\beta.
    \end{align*}
    Thus, $\strats$ is a $C\beta$-approximate Nash equilibrium.
\end{proof}

\begin{lemma}[$\beta$-smoothed equilibria converge to Nash equilibria] \label{lem:convergence-to-Nash}
    Let $\domains$ be a convex, joint strategy space with a metric $\rho$. Let $(\domains, \utilities)$ be an $N$-player game with continuous utilities, and $\regs$ be a set of strictly convex and bounded regularizers. Suppose that the game has a unique Nash equilibrium $\NE$. Then:
    \[\lim_{\beta \to 0}\, \rho(\SE, \NE) = 0.\]
\end{lemma}

\begin{proof}[Proof of \Cref{lem:convergence-to-Nash}]
    This follows from \Cref{lem:smooth-to-Nash} and \Cref{lem:approx-nash-convergence}.
\end{proof}

\begin{lemma}[Convergence to Nash equilibria] \label{lem:approx-nash-convergence}
    Let $\domains$ be a joint strategy space with a metric $\rho$, and suppose that $(\domains, \utilities)$ is an $N$-player game with continuous utilities. Let $(\strats^{\epsilon_k})_{k \in \NN}$ be any sequence of $\epsilon_k$-Nash equilibria where $\epsilon_k \to 0$. If $\NE$ is a limit point of the sequence, then $\NE$ is a Nash equilibrium.
\end{lemma}

\begin{proof}
   The result follows since:
    \begin{enumerate}
        \item Let $\mathbf{X}^\epsilon$ denote the set of $\epsilon$-Nash equilibria. It is a closed set for all $\epsilon > 0$.
        \item The $\epsilon$-Nash equilibria satisfies $\mathbf{X}^\epsilon \subset \mathbf{X}^{\epsilon'}$ whenever $\epsilon \leq \epsilon'$. Also, the set of Nash equilibria is given by:
        \[\mathbf{X}^* = \bigcap_{\epsilon > 0} \mathbf{X}^\epsilon.\]
        \item For each $\epsilon$, the sequence $\strats^{\epsilon_k}$ eventually remains in $\mathbf{X}^\epsilon$, and so $\NE \in \mathbf{X}^\epsilon$ for all $\epsilon$.
    \end{enumerate}
    This implies that $\NE$ is Nash.
\end{proof}

\section{Strategic equivalence and games in canonical form}

For many results of this work, we may often restrict our consideration to games in \emph{canonical form}:

\begin{definition}[Canonical form of a multilinear game] \label{def:canonical}
    Let $(\domains, \utilities)$ be an $N$-player multilinear game, where the joint strategy space includes the origin $\mathbf{0} \in \domains$. We say that the game is in \emph{canonical form} if the utilities have no non-strategic component and if the origin is a Nash equilibrium, with $\NE = \mathbf{0}$.
\end{definition}

We can do this without loss of generality because the solution concepts we are concerned with all respect strategic equivalence, namely strategic Pareto optimality (\Cref{def:strategic-pareto-optimality}) and uniform stability (\Cref{def:uniform-stability}). Therefore, it is enough to prove them for games that are purely strategic (i.e.\ they have no non-strategic component; see \Cref{def:strategic-component}). Furthermore, in multilinear games, we can assume that the Nash equilibrium is set at the origin $\NE = \mathbf{0}$. We may do this without loss of generality since the translation of a multilinear polynomial remains a multilinear polynomial.

\begin{lemma} \label{lem:strategic-equivalence}
    The dynamics \eqref{eqn:dynamics} are uncoupled and are preserved across strategic equivalence classes.
\end{lemma}

\begin{proof}
    The dynamics depend only on the strategic components of the game.
\end{proof}

\section[Proof of Lemma \ref*{lem:pd-stretch}]{Proof of \Cref{lem:pd-stretch}}
\label{sec:pd-stretch}

\pdstretch*
\begin{proof}
    ($\Rightarrow$). Suppose that $u^\top v > 0$. There are two cases:
    \begin{enumerate}
        \item If $u$ and $v$ are linearly dependent, then we can let $H = \lambda I$ where $\lambda = \|u\| / \|v\|$, so that $Hv = u$.
        \item Otherwise, $u$ and $v$ span two dimensions. We can choose an orthonormal basis $B$ such that: 
        \[B^\top u = \alpha_1 e_1 + \alpha_2 e_2 \quad \textrm{and}\quad B^\top v = \beta_1 e_1 + \beta_2 e_2, \qquad \alpha_1, \alpha_2, \beta_1, \beta_2 > 0.\]
        This is possible because $u^\top v > 0$. Now, let $\Lambda = \mathrm{diag}(\alpha_1/\beta_1, \alpha_2/\beta_2, 1,\ldots, 1)$ and set:
        \[H = B \Lambda B^\top,\]
        so that by construction, $u = Hv$.
    \end{enumerate}
    ($\Leftarrow$). Suppose that $u = Hv$ for some positive-definite matrix $H$. Then $u^\top v = v^\top H v > 0$.
\end{proof}

\section[Proof of Lemma \ref*{lem:gradient-sbr}]{Proof of \Cref{lem:gradient-sbr}}
\label{sec:gradient-sbf}

\gradientsbr*

\begin{proof}
    The map $x_n' \mapsto f_n(x_n', \strats_{-n}) - \beta_n h_n(x_n)$ is a strictly concave function over a convex set $\Omega_n$. And so, it has a unique maximizer, implying that $\Phi^\beta(\strats)_n$ is well-defined.
    
    Given any $\strats$ that maps to a fully mixed strategy $\betaBRs(\strats) \in \mathrm{int}(\domains)$, we can apply the Lagrange multiplier theorem to express $\betaBRs(\strats)_n$ for each $n \in [N]$ as the stationary point to the Lagrangian:
    \[\cL_i(x_n', \lambda; \strats) = f_n(x_n', \strats_{-n}) -  \beta h_n(x_n') - \lambda \ind^\top x_n',\]
    where $\ind \in \RR^{d_n}$ is the all-ones vector. This implies that $\Phi^\beta(\strats)_n$ is the unique stationary point satisfying:
    \[0 = \nabla_{x_n'} \cL_n\big(x_n', \lambda;\strats\big) = \nabla_{x_n'} \Big(f_n(x_n', \strats_{-n}) - \beta h_n(x_n') \Big) - \lambda \ind.\]
    Let $\Pi_n = \mathrm{I} - \frac{1}{d_n} \ind \ind^\top$ be the linear projection and $\Psi_n : \Omega_n \times \domains \to T\Omega_n$ be the map:
    \[\Psi_n(x_n', \strats) = \nabla_{n} \Big(f_n(x_n', \strats_{-n}) - \beta h_n(x_n')\Big).\]
    Then, $x_n' = \Phi_\beta(\strats)_n$ is the unique solution to  $\Psi_n\big(x_n', \strats\big) = 0$. By the implicit function theorem, we obtain the gradient of $\betaBRs$. The blocks on the diagonals are zero because $f_n$ is multilinear, while $\nabla_{n}^2h_n$ is positive definite on $\Omega_n$ since $h_n$ is strictly convex. 

    The smoothness of $\nabla \betaBRs$ follows from Cramer's rule: $\Hs$ is smooth and positive definite.
\end{proof}

\section[Proof of Theorem \ref*{thm:meaning-local}]{Proof of \Cref{thm:meaning-local}} \label{sec:meaning-local}

\Cref{thm:meaning-local} considers an $N$-player normal-form game with an interior equilibrium that is locally uniformly stable. To show that the equilibrium is locally strategically Pareto optimal, we may restrict ourselves to a sufficiently small open set relative to centered at the equilibrium. This restriction is a multilinear game (\Cref{lem:normal-to-multilinear}), and by taking an affine change of coordinates and by removing the non-strategic component, we may assume that the game is in canonical form (\Cref{def:canonical}). \Cref{thm:meaning-local} follows immediately from:

\begin{theorem}
Let $(\domains, \utilities)$ be a multilinear game in canonical form, so that $\mathbf{0}$ is a Nash equilibrium. Suppose that $\mathbf{0}$ is locally uniformly stable. Then, $\mathbf{0}$ is locally weakly Pareto optimal.
\end{theorem}

\begin{proof} 
Without loss of generality, suppose that there is an open box $U$ centered at $\mathbf{0}$ on which the conditions of uniform stability hold. We show that for each $\mathbf{x} \in U$, there is some player $k = k(\mathbf{x}) \in [N]$ such that
\begin{equation} \label{eqn:individual-non-improvement}
f_k(\mathbf{x}) \leq 0.
\end{equation}
Thus, the joint strategy $\mathbf{x}$ does not improve the utility of player $k$. Since this is true for arbitrary $\mathbf{x} \in U$, the origin $\mathbf{0}$ is locally weakly Pareto optimal.

If $x_i=0$ for some $i \in [N]$, then $f_i(\mathbf{x}) = x_i^\top \nabla_i f_i(\mathbf{x}) = 0$, since the utilities are multilinear and purely strategic. In this case, $\mathbf{x}$ cannot strictly improve every player. In the remainder, we fix any $\mathbf{x} \in U$ such that $x_i \ne 0$ for all $i \in [N]$.

Consider a restricted game where each player $i \in [N]$ is restricted to playing along the one-dimensional segment $\{t_i x_i : t_i \in [0,1]\}$. Let $\mathbf{t} = (t_1,\ldots, t_N) \in [0,1]^N$ and define the restricted game by
\[f^\mathbf{x}_i(\mathbf{t}) = f_i(\mathbf{t} \odot \mathbf{x}) =  f_i(t_1 x_1,\ldots, t_N x_N),\]
where $\mathbf{t} \odot \mathbf{x} = (t_1 x_1,\ldots, t_N x_N)$. As the utilities are multilinear and purely strategic at $\mathbf{0}$, we can write
\[f_i^\mathbf{x}(\mathbf{t}) = t_i g_i(\mathbf{t})\qquad \textrm{where }g_i(\mathbf{t}) := \partial_{t_i}f_i^\mathbf{x}(\mathbf{t}).\]
The chain rule gives $g_i(\mathbf{t}) = x_i^\top \nabla_i f_i(\mathbf{t} \odot \mathbf{x})$. Let $G \equiv (g_1,\ldots, g_N) : [0,1]^N \to \mathbb{R}^N$. Because $\mathbf{0}$ is a Nash equilibrium, we have that $g_i(\mathbf{0}) = 0$.

\begin{lemma} \label{lem:P0-function}
    The function $-G$ is a $\mathrm{P}_0$-function.
\end{lemma}

By definition of a $\mathrm{P}_0$-function, this means that for all $\mathbf{t} \in [0,1]^N$, there is an index $k = k(\mathbf{t})$ such that the inequality holds $$-f_k^\mathbf{x}(\mathbf{t}) =  (t_k - 0) [- g_k(\mathbf{t}) + g_k(\mathbf{0})] \geq 0.$$
In particular, when $\mathbf{t} = (1,\ldots, 1)$, this shows \Cref{eqn:individual-non-improvement}.

\begin{proof}[Proof of \Cref{lem:P0-function}]
The claim will follow from two claims. The first claim is specific to the game setting of this theorem, while the second claim is a general linear-algebraic statement.
\begin{itemize}
    \item \textbf{Claim 1.} The game Jacobian $J^\mathbf{x}(\mathbf{t})$ of the restricted game $(f_1^\mathbf{x},\ldots, f_N^\mathbf{x})$ is uniformly stable for all $\mathbf{t} \in [0,1]^N$. That is, for all positive-definite diagonal matrix $\Lambda \in \mathbb{R}^{N \times N}$, the matrix $\Lambda^{-1} J^\mathbf{x}$ has purely imaginary eigenvalues, where the $ij$-entry is $$J^\mathbf{x}_{ij}(\mathbf{t}) = \partial_{t_j} g_i(\mathbf{t}).$$
    \item \textbf{Claim 2.} Let $J \in \mathbb{R}^{N \times N}$ be any matrix such that $\mathrm{spec}(\Lambda^{-1} J) \subset i \mathbb{R}$ for all positive-definite diagonal matrices $\Lambda \in \mathbb{R}^{N \times N}$. Then, $J$ and $-J$ are $\mathrm{P}_0$-matrices.
\end{itemize}
These two claims show that the negative game Jacobian $-J^\mathbf{x}(\mathbf{t})$ is a $\mathrm{P}_0$-matrix on all of $\mathbf{t} \in[0,1]^N$. Thus, by Corollary 5.3 of \cite{more1973pfunctions} (\Cref{thm:P0-function}), the function $-G$ is a $\mathrm{P}_0$-function.

\paragraph{Proof of Claim 1}
Fix any $\mathbf{t} \in [0,1]^N$. The restricted game Jacobian $J^\mathbf{x}$ is related to the original Jacobian $\mathbf{J}$ through the equation
\[J_{ij}^\mathbf{x}(\mathbf{t}) = x_i^\top J_{ij}(\mathbf{t} \odot \mathbf{x})x_j,\] 
which follows from applying the chain rule to $\partial_{t_j} g_i$. Because $\mathbf{J}$ is uniformly stable on the open box $U$ containing $\mathbf{x}$,  it is uniformly stable at $\mathbf{t} \odot \mathbf{x}$. In the remainder, we will suppress $\mathbf{t}$ and $\mathbf{x}$ to reduce notation.

Fix any positive-definite diagonal matrix $\Lambda  = \mathrm{diag}(\lambda_1,\ldots, \lambda_N)$. For each $\varepsilon \geq 0$ and $i \in [N]$, define
\[K_i^\varepsilon = \lambda_i^{-1} x_ix_i^\top + \varepsilon I,\]
and define the following block matrix $\mathbf{K}^\varepsilon = \mathrm{diag}(K_1^\varepsilon,\ldots, K_N^\varepsilon)$. Let $\mathbf{K} = \mathbf{K}^0$ for short. Whenever $\varepsilon > 0$, the matrix $(\mathbf{K}^\varepsilon)^{-1}$ is positive definite and we thus have by uniform stability of $\mathbf{J}$ that 
\[\mathrm{spec}\big(\mathbf{K}^\varepsilon\mathbf{J}\big) \subset i \mathbb{R}.\]
The set $i \mathbb{R} \subset \mathbb{C}$ is a closed. As $\mathbf{K}^\varepsilon \to \mathbf{K}$ as $\varepsilon \to 0$, by the continuity of eigenvalues (\Cref{thm:eigenvalue-perturbation}), we have
\[\mathrm{spec}\big(\mathbf{K}\mathbf{J}\big) \subset i \mathbb{R}.\]
For each $i \in [N]$, we lift $x_i \in V_i$ into $\mathbf{V} = V_1 \times \dotsm \times V_N$ and denote the vector by
\[\mathbf{b}_i = 0 \oplus \dotsm \oplus x_i \oplus 0 \dotsm \oplus 0.\]
Because we assumed that $x_i$ is non-zero, neither is $\mathbf{b}_i$. Let $\mathbf{B}$ be any orthogonal basis of $\mathbf{V}$ where the first $N$ columns are $\mathbf{b}_1,\ldots, \mathbf{b}_N$. Under this choice of basis, the following matrix has the block form 
\[\mathbf{B}^{-1}\mathbf{K}\mathbf{J}\mathbf{B} = \begin{pmatrix} \Lambda^{-1} J^\mathbf{x} & * \\ 0 & 0\end{pmatrix}.\]
The upper left block is $\Lambda^{-1} J^\mathbf{x}$ because $\mathbf{b}_i$ and $\mathbf{b}_j$ respectively only detect directions in the $V_i$ and $V_j$ components of $\mathbf{V}$. And so, the only contribution must come from the $ij$-block of $\mathbf{K}\mathbf{J}$, which is 
\[(\mathbf{K}\mathbf{J})_{ij} = \lambda_i^{-1} x_i x_i^\top J_{ij}.\] 
It follows that the $ij$-entry is $\lambda_i^{-1} x_i^\top J_{ij} x_j = (\Lambda^{-1} J^\mathbf{x})_{ij}$. The lower blocks are zero because the rowspace of $\mathbf{K}$ is completely contained in $\mathrm{span}(\mathbf{b}_1,\ldots, \mathbf{b}_N)$.

The matrix $\mathbf{B}^{-1} \mathbf{K}\mathbf{J}\mathbf{B}$ is similar to $\mathbf{K}\mathbf{J}$, so they have precisely the same eigenvalues. Furthermore, $\mathbf{B}^{-1} \mathbf{K}\mathbf{J}\mathbf{B}$ is upper triangular, implying the chain of inclusions
\[\mathrm{spec}(\Lambda^{-1} J^\mathbf{x}) \subset \mathrm{spec}(\mathbf{B}^{-1} \mathbf{K}\mathbf{J}\mathbf{B}) = \mathrm{spec}(\mathbf{KJ}) \subset i \mathbb{R}.\]
As $\Lambda \in \mathbb{R}^{N \times N}$ was an arbitrary positive-definite diagonal matrix, $J^\mathbf{x}$ is uniformly stable.

\paragraph{Proof of Claim 2}
Let $S \subset [N]$ be a subset of indices. Let $M = J[S^c]$ denote the principal submatrix constructed by removing the rows and columns with index in $S$. 

We first show that $\mathrm{spec}(M) \subset i \mathbb{R}$. That is, $M$ has purely imaginary eigenvalues. To do so, define $D_S^\varepsilon$ to be the diagonal matrix where 
\[D_S^\varepsilon(i,i) = \begin{cases}1 & i \notin S\\\varepsilon & i \in S.\end{cases}\]
Then, $D^\varepsilon_S J$ has purely imaginary eigenvalues. Moreover, $D^\varepsilon_S J \to D^0_S J$, by the continuity of eigenvalues (\Cref{thm:eigenvalue-perturbation}), the eigenvalues of $D_S^0 J$ are also purely imaginary. Thus:
\[\mathrm{spec}(M) \subset \mathrm{spec}(D_S^0 J)\subset i\mathbb{R}.\]
The claim that $J$ is $\mathrm{P}_0$ follows with two cases:
\begin{enumerate}
    \item If the dimension of $M$ is odd, then at least one eigenvalue is zero. The determinant is zero.
    \item If the dimension of $M$ is even, then every eigenvalue has a conjugate pair $i\lambda, -i\lambda$ where $\lambda \in \mathbb{R}$. The product is non-negative $\lambda^2 \geq 0$. Thus, the determinant is non-negative.
\end{enumerate}
Since $J[S^c]$ has purely imaginary eigenvalues if and only if $-J[S^c]$ has purely imaginary eigenvalues, the same proof also shows that $-J$ is $\mathrm{P}_0$.
\end{proof}
\end{proof}

\begin{theorem}[Continuity of eigenvalues, Theorem 1 of \cite{elsner1985optimal}] \label{thm:eigenvalue-perturbation}
    Let $A, B \in \CC^{m\times m}$ have spectra $\{\lambda_1,\ldots, \lambda_m\}$ and $\{\mu_1,\ldots, \mu_m\}$, respectively. Then:
    \[\max_{j \in [m]} \min_{i \in [m]}\,\big|\lambda_i - \mu_j\big| \leq \big(\|A\|_F + \|B\|_F\big)^{(m-1)/m} \cdot \big\| A - B\big\|_F^{1/m}.\]
\end{theorem}

\section[Proof of Proposition \ref*{prop:non-convergence-general}]{Proof of \Cref{prop:non-convergence-general}}

\label{sec:proof-non-convergence}

\inapprox*

\begin{proof}
    In the following, let $\Hs(\strats)$ denote the block-diagonal matrix with $H_n(\strats) = \nabla_n^2 h_n(x_n)$ given a set of smooth and strictly convex regularizers $\regs$.
    
    As $\NE$ is not uniformly stable, we can choose $\regs$ such that the eigenvalues of $\Hs^{-1}(\NE) \Jacobian(\NE)$ are not all purely imaginary (such a choice exists by \Cref{lem:choice-regularizer}). In fact, it must have an eigenvalue $\lambda^*$ with positive real part $\Re(\lambda^*) > c > 0$. This is because the matrix $\Hs^{-1}(\NE)\Jacobian(\NE)$ has zero trace---the diagonal blocks of the matrix $H_n(\NE) J_{nn}(\NE)$ in multilinear games is zero---and so, the sum of all eigenvalues must also be zero. 
    
    The existence of eigenvalues with positive real parts, bounded away from zero, extends to an open region around $\NE$ by continuity. In particular, the matrix-valued function $\strats \mapsto \Hs^{-1}(\strats) \Jacobian(\strats)$ is continuous, so by the continuity of eigenvalues (\Cref{thm:eigenvalue-perturbation}), there is an open set $U$ around $\NE$ so that for all $\strats \in U$, each matrix $\Hs^{-1}(\strats) \Jacobian(\strats)$ also has an eigenvalue $\lambda(\strats)$, where its real part is bounded away from zero: 
    \[\Re\big(\lambda(\strats)\big) > c/2 > 0.\]
    We now show that the dynamics do not stabilize to $\NE$. That is, either the sequence of $\beta$-smoothed equilibria $\SE$ does not converge to $\NE$ as $\beta$ goes to zero, or the dynamics become unstable around $\SE$. If the sequence does not converge, then there is nothing to do. So, without loss of generality, we may assume that there exists some sufficiently small $\beta_0 < 2/c$ so that all $\beta$-smoothed equilibria remain in $U$ when $\beta < \beta_0$:
    \[\forall \beta < \beta_0,\qquad \SE \in U.\]
    
    We can now analyze the linear stability of $\SE$ under the $(\betaBRs, \eta)$-averaging dynamics \eqref{eqn:dynamics}. Its Jacobian is:
    \[\nabla\Big((1 - \eta)\, \strats + \eta\, \betaBRs(\strats)\Big) = (1 - \eta)\, \mathbf{I} + \frac{\eta \,\Hs^{-1}(\strats)\Jacobian(\strats)}{\beta},\]
    where $\mathbf{I}$ is the identity matrix, making use of the gradient computation \Cref{lem:gradient-sbr}. The eigenvalues are:
    \[\hspace{3em}(1 - \eta) + \frac{\eta \lambda}{\beta}, \qquad\textrm{where } \lambda \in \mathrm{spec}\big(\Hs^{-1}(\SE)\Jacobian(\SE)\big).\]
    Whenever $\beta < \beta_0$, we have that $\SE \in U$, so that there is at least one $\lambda$ whose real part is at least $c/2$. Moreover, as $\beta_0 < 2/c$, it follows that the modulus is bounded below for all $\eta \in (0,1)$:
    \[\Big|(1 - \eta) + \frac{\eta \lambda}{\beta}\Big| > 1 - \eta + \frac{\eta\, c}{2\beta} > 1.\]
    The fixed point $\SE$ is unstable when $\beta < \beta_0$ as the Jacobian of the dynamics has an eigenvalue whose modulus is greater than 1. This follows from Lyapunov's indirect method \cite[Theorem 3.3]{bof2018lyapunov}.
\end{proof}

\section[Proof of Theorem \ref*{thm:main-convergence}]{Proof of \Cref{thm:main-convergence}}
\label{sec:proof-main-convergence}

\mainconv*

\begin{proof}
We first introduce some notation. Recall the $\beta$-smoothed learning dynamics are given by
\[
\hspace{5em}\strats(t+1) = T_{\beta, \eta} (\strats),\qquad \textrm{where }T_{\beta, \eta} (\strats) = (1-\eta) \strats + \eta \Phi^\beta (\strats).
\]
Also recall from \Cref{lem:gradient-sbr} that
\[\nabla \betaBRs(\strats) = \frac{1}{\beta} \Hs(\strats)^{-1} \Jacobian(\strats).\] 
Define the linear operators $A_\beta$ and $B_\beta$ as follows:
\[
A_\beta = \Hs(\SE)^{-1} \Jacobian(\SE) \qquad \textrm{and} \qquad B_\beta = DT(\SE) = (1 - \eta) \mathbf{I} + \frac{\eta}{\beta} A_\beta,
\]
so that $B_\beta$ linearizes $T_{\beta,\eta}$ about the fixed point $\SE$. 

By uniform stability, the spectrum of $A_\beta$ is purely imaginary. Let $W_\beta = \rho(A_\beta)$ be the spectral radius, and choose step size $\eta = \smash{\frac{1}{2}\frac{\beta^2}{\beta^2+W_\beta^2}}$. The eigenvalues of $B_\beta$ are of the form: $1-\eta+\frac{\eta}{\beta} i \omega$.
Plugging in this choice of $\eta$, we see that the spectral radius $\rho(B_\beta) \leq e^{-3\eta/4}$ is bounded:
\begin{align*}
     \rho(B_\beta) &\leq \left|1 - \eta + i \frac{\eta\omega}{\beta}\right|
    \\&= \sqrt{1 - 2\eta + \eta^2\left(1 + \frac{\omega^2}{\beta^2}\right)}
    \\&\leq \exp\left(- \eta + \frac{1}{2}\eta^2 \left(1 + \frac{\omega^2}{\beta^2}\right)\right)
    \\&\leq \exp\left(- \frac{3\eta}{4}\right).
\end{align*}

\paragraph{Construction of a Lyapunov function}
Fix the constant $q=e^{-5\eta/8}$, so that $\rho(B_\beta)<q<e^{-\eta/2}$. Since we have the bound $\rho(q^{-1} B_\beta) < 1$, the following is a convergent series:
\begin{equation}\label{eqn:P-beta-def}
P_\beta = \sum_{k=0}^\infty q^{-2k} (B_\beta^k)^\top B_\beta^k.
\end{equation}
The operator $P_\beta \succ 0$ is positive definite, since it is a series of positive-definite operators: $B_\beta$ is full-rank since all of its eigenvalues have real part $1 - \eta$.
The operator $P_\beta$ also solves the discrete Lyapunov equation:
\[
B_\beta^\top P_\beta B_\beta - q^2 P_\beta + q^2 \mathbf{I} = 0,
\]
which can be directly verified by a little bit of algebra.
Define the Lyapunov function:
\[\|\strats - \strats^\beta\|_{P_\beta}^2 = (\strats - \strats^\beta)^\top P_\beta (\strats - \strats^\beta),\]
and let $\mathbb{B}_{P_\beta}(\strats, r)$ denote the $r$-ball with respect to the $\|\cdot\|_{P_\beta}$-norm around $\strats$. Let $\|\cdot\|_\mathrm{op}$ denote the operator norm induced by $\|\cdot\|_{P_\beta}$ on the space of linear endomorphisms. In particular, we have that
\begin{equation}\label{eqn:Bbeta-bound}
    \|B_\beta\|_{\mathrm{op}}^2 = \sup_{\|\mathbf{u}\|_{P_\beta} = 1} \,\mathbf{u}^\top B_\beta^\top P_\beta B_\beta \mathbf{u} =  \sup_{\|\mathbf{u}\|_{P_\beta} = 1} \, q^2\mathbf{u}^\top P_\beta \mathbf{u}- q^2 \mathbf{u}^\top \mathbf{u} < q^2,
\end{equation}
where the second inequality uses the discrete Lyapunov equation.

\paragraph{Perturbation bounds} By choosing a sufficiently small radius $R_\beta > 0$, the closed ball $\overline{\mathbb{B}}_{P_\beta}(x^\beta, R_\beta)$ is contained in the domain. The map $\strats \mapsto \mathbf{H}(\strats)^{-1} \mathbf{J}(\strats)$ is smooth (\Cref{lem:gradient-sbr}), so it is $L$-Lipschitz for some constant $L \equiv L_{A,\beta,R_\beta}$. Define the smaller radius $0 < r_\beta < R_\beta$, which scales linearly with $\beta$, as follows:
\[
r_\beta \leq \min\left\{R_\beta, \frac{\beta}{\eta L} (e^{-\eta/2} - q) \right\}.
\]
We obtain the bound:
\begin{align}
\hspace{4em} \|DT(\strats) - B_\beta\|_\mathrm{op} \leq e^{-\eta/2} - q, \qquad\forall \strats \in \mathbb{B}_{P_\beta}(x^\beta, r_\beta).
\label{eq:DT_bound}
\end{align}
This is verified by the following computation:
\begin{align*}
    \|DT(\strats) - B_\beta\|_\mathrm{op} &= \frac{\eta}{\beta} \left\|\Hs(\strats)^{-1} \Jacobian(\strats) - \Hs(\SE)^{-1} \Jacobian(\SE)\right\|_\mathrm{op}
    \\&\leq \frac{\eta}{\beta} L \|\strats - \SE\|_{P_\beta} 
    \\&\leq \frac{ \eta L r_\beta}{\beta}.
\end{align*}
By the triangle inequality, \eqref{eqn:Bbeta-bound} and \eqref{eq:DT_bound} implies the upper bound
\begin{equation}\label{eqn:DT-contract}
\hspace{4em}\|DT(\strats)\|_{P_\beta} \leq e^{-\eta/2}, \qquad \forall \strats\in \mathbb{B}_{P_\beta}(x^\beta, r_\beta).
\end{equation}

\paragraph{Contraction of dynamics}
We now use \eqref{eqn:DT-contract} to show that the dynamics $\strats \mapsto T_{\beta,\eta}(\strats) - \SE$ is a contraction with respect to $\|\cdot \|_{P_\beta}$. Fix any $\strats\in \mathbb{B}_{P_\beta}(\strats^\beta, r_\beta)$. Since $\mathbb{B}_{P_\beta}(\strats^\beta, r_\beta)$ is convex, the entire line segment
\[
\strats^\beta+s(\strats-\strats^\beta),
\qquad s\in[0,1],
\]
is contained in $\mathbb{B}_{P_\beta}(\strats^\beta, r_\beta)$. Applying the integral form of the mean-value theorem, we obtain:
\[
T_{\beta,\eta}(\strats)-T_{\beta,\eta}(\strats^\beta)
=
\int_0^1
DT_{\beta,\eta}
\left(
\strats^\beta+s(\strats-\strats^\beta)
\right)
(\strats-\strats^\beta)\,ds.
\]
Since $\strats^\beta$ is a fixed point of $T_{\beta,\eta}$, we have
\[
\begin{aligned}
\|T_{\beta,\eta}(\strats)-\strats^\beta\|_{P_{\beta}}
&=
\|T_{\beta,\eta}(\strats)-T_{\beta,\eta}(\strats^\beta)\|_{P_{\beta}} \\
&\le
\int_0^1
\left\|
DT_{\beta,\eta}
\left(
\strats^\beta+s(\strats-\strats^\beta)
\right)
\right\|_{P_{\beta}}
\,ds\,
\|\strats-\strats^\beta\|_{P_{\beta}} \\
&\le
e^{-\eta/2}
\|\strats-\strats^\beta\|_{P_{\beta}}.
\end{aligned}
\]
Consequently, the map $T_{\beta,\eta}$ is a contraction on
$\mathbb{B}_{P_\beta}(\strats^\beta, r_\beta)$ with contraction factor
$e^{-\eta/2}$. 

Let $\strats(0)\in \mathbb{B}_{P_\beta}(\strats^\beta, r_\beta)$. Repeated applications of the above inequality yields
\[
\|\strats(t)-\strats^\beta\|_{P_{\beta}}
\le
e^{-\eta t/2}
\|\strats(0)-\strats^\beta\|_{P_{\beta}}.
\]
Let $\kappa(P_\beta)$ be the condition number $\lambda_\mathrm{max}(P_\beta)/\lambda_\mathrm{min}(P_\beta)$. Then, the above inequality may be converted into a contraction with respect to the Euclidean distance:
\[
\|\strats(t)-\strats^\beta\|_2
\le
\kappa(P_\beta) e^{-\eta t/2}
\|\strats(0)-\strats^\beta\|_2.
\]
The result follows from applying the following lemma bounding $\kappa(P_\beta)$.
\begin{lemma}[Condition number bound]\label{lem:condition-number}
    The condition number is bounded by
    \[\kappa(P_\beta) \leq (4d)! \sum_{\ell=0}^{d-1} \left(1 + \frac{\eta K}{\beta}\right)^{2\ell} \frac{e^{5 \eta \ell / 4}}{(1 - e^{-\eta / 4})^{2\ell + 1}},\]
    where $K = \|\mathbf{H}^{-1}(\SE) \mathbf{J}(\SE)\|_\mathrm{op}$. In particular, if $\eta = \Theta(\beta^2)$ and $\beta$ is sufficiently small,
    \[ \log \kappa(P_\beta) = O\left(d \log d + d \log \frac{1}{\beta}\right).\]
\end{lemma}
\begin{proof}[Proof of \Cref{lem:condition-number}]
    For convenience, recall
    \[P_\beta = \sum_{k=0}^\infty q^{-2k} (B_\beta^k)^\top B_\beta^k.\tag{\ref{eqn:P-beta-def}}\]
    Since $P_\beta$ is a infinite series of positive-definite matrices, where the $k = 0$ term in the series is the identity matrix $I$. It follows that $\lambda_\mathrm{min}(P_\beta) \geq 1$, and that $\kappa(P_\beta) \leq \lambda_\mathrm{max}(P_\beta)$.

    We now compute a crude upper bound on $\lambda_\mathrm{max}(P_\beta)$. Let $B_\beta = Q \Gamma Q^*$ be the Schur decomposition of $B_\beta$, where $Q$ is unitary and $\Gamma$ is upper triangular. In particular, let $\Gamma = \Lambda + N$ where $\Lambda$ is diagonal matrix consisting of the eigenvalues of $B_\beta$ and $N$ is a strictly upper triangular (hence nilpotent) matrix. Then:
    \[\|B_\beta^k\|_\mathrm{op} \leq \sum_{\ell=0}^{\min\{k, d-1\}} \binom{k}{\ell} \|\Lambda\|_\mathrm{op}^{k-\ell} \cdot \|N\|_\mathrm{op}^\ell,\]
    as shown in Lemma 2.1 of \cite{davies2003schur}. Let $M = 1 + \frac{\eta}{\beta} K$ so that $\|\Lambda\|_\mathrm{op} = \rho(B_\beta) \leq M$ and
    \[\|N\|_\mathrm{op} = \|Q(\Lambda + N)Q^* - Q \Lambda Q^*\|\leq  \|B_\beta\|_\mathrm{op} + \|\Lambda\|_\mathrm{op} \leq 2M,\]
    we obtain the bound for all $k \geq 0$ that
    \[\|B_\beta^k\|_\mathrm{op} \leq \sum_{\ell=0}^{d-1} \binom{k}{\ell} (2M)^\ell \rho(B_\beta)^{k-\ell}.\]
    where $\binom{k}{\ell} = 0$ whenever $\ell > k$. Let $\rho \equiv \rho(B_\beta)$ for short. Cauchy-Schwarz gives the bound
    \[\|B_\beta^k\|_\mathrm{op}^2 \leq d \sum_{\ell=0}^{d-1} \binom{k}{\ell}^2 (2M)^{2\ell} \rho^{2(k-\ell)}.\]
    This implies
    \begin{align*}
        \|P_\beta\|_\mathrm{op} &\leq \sum_{k=0}^\infty q^{-2k} \|B_\beta^k\|_\mathrm{op}^{2}
        \\&\leq d \sum_{k=0}^\infty \sum_{\ell=0}^{d-1} \left(\frac{2M}{q}\right)^{2\ell} \binom{k}{\ell}^2 \left(\frac{\rho}{q}\right)^{2(k - \ell)}
        \\&\leq d\sum_{\ell = 0}^{d-1} \left(\frac{2M}{q}\right)^{2\ell}\sum_{n=0}^\infty \binom{n+\ell}{\ell}^2 \left(\frac{\rho}{q}\right)^{2n}.
    \end{align*}
    Recall that $\rho \leq e^{-3 \eta /4}$ and $q = e^{-5\eta / 8}$, so that $z = (\rho/q)^2 < e^{-\eta / 8}$. Let $A_m$ denote the $m$th Eulerian polynomial \citep[Equation 1.36]{stanley2011enumerative}. The inner summation can be bounded as follows:
    \begin{align*}
        \sum_{n=0}^\infty \binom{n + \ell}{\ell}^2 z^n &= \sum_{n=0}^\infty  \prod_{j = 1}^\ell \left(1 + \frac{n}{j}\right)^2 z^n
        \\&\leq \sum_{n=0}^\infty (1 + n)^{2\ell} z^n
        \\&= \frac{1}{z}\sum_{m=0}^\infty m^{2\ell} z^m 
        \\&\overset{(i)}{=} \frac{1}{z}\frac{A_{2\ell}(z)}{(1 - z)^{2\ell + 1}}
        \\&\overset{(ii)}{\leq} \frac{(2\ell)!}{(1 - z)^{2\ell + 1}}
    \end{align*}
    where (i) applies Proposition 1.4.4 of \cite{stanley2011enumerative}, and (ii) uses the fact that $A_{2\ell }(z)/z \leq A_{2\ell}(1) = (2\ell)!$ which is explained around Equation 1.36 of \cite{stanley2011enumerative}. Returning to bounding $\|P_\beta\|_\mathrm{op}$, we obtain
    \begin{align*}
        \|P_\beta\|_\mathrm{op} \leq (4d)! \sum_{\ell=0}^{d-1} \left(1 + \frac{\eta K}{\beta}\right)^{2\ell} \frac{e^{5 \eta \ell / 4}}{(1 - e^{-\eta / 4})^{2\ell + 1}}.
    \end{align*}
\end{proof}
\end{proof}

\section{Beyond interior Nash equilibria} \label{sec:beyond-interior}

The convergence result \Cref{thm:main-convergence} considered only interior Nash equilibria. This section extends this result to Nash equilibria when they are on the boundary of the simplex.

The following lemma shows that if $\NE$ is an isolated, non-interior equilibrium, then the game can be locally \emph{reduced} by removing actions that are not in the support of $\NE$. These actions turn out to be strictly dominated, and so the game-theoretic equilibrium $\NE$ is preserved after removing dominated strategies.

\begin{lemma} \label{lem:local-domination}
    Let $(\domains, \utilities)$ be an $N$-player normal-form game with a quasi-strong equilibrium $\NE \in \domains$. There exists a neighborhood $U$ containing $\NE$ such that actions not supported by $\NE$ are strictly dominated. That is, there exists some $\epsilon > 0$ such that for each player $n \in [N]$ and for all $\strats \in U$ and $i \notin \mathrm{supp}(x_n^*)$,
    \[f_n(e_{n,i}; \strats_{-n}) < \max_{j \in [k_n]}\, f_n(e_{n,j}; \strats_{-n}) - \epsilon.\]
\end{lemma}
\begin{proof}
    As $\NE$ is quasi-strong, for each player $n \in [N]$, there is some $\epsilon_n > 0$ such that whenever a pure strategy $i \in [k_n]$ is not supported, meaning that $i \notin \mathrm{supp}(\NE)$, then it is strictly dominated:
    \[f_n(e_{n,i}; \NE_{-n}) < \max_{j \in [k_n]} f_n(e_{n,j}; \NE_{-n}) - \epsilon_n.\]
    This is because the set of pure strategies that are best-responses to $\NE$ are precisely those in $\mathrm{supp}(x_n^*)$. As this is a finite set of alternatives, there must be some positive gap between those that maximize utility and those that do not. And as $f_n$ is continuous, for a sufficiently small ball $B_n$ around $\NE$, the same domination condition remains true, although perhaps with a smaller gap. That is, when $i \notin \mathrm{supp}(\NE)$,
    \[\forall \strats \in B_n,\qquad f_n(e_{n,i}; x_{-i}) < \max_{j \in[k_n]}\, f_n(e_{n,j}; \strats_{-n}) - \epsilon_n/2.\hspace{30pt}\]
    The result follows by letting $U$ be the intersection $B_1 \cap \dotsm \cap B_N$ and $\epsilon = \min_{n \in [N]}\, \epsilon_n/2$.
\end{proof}

\begin{corollary}[Reduction to interior equilibria]
    \label{cor:reduction-eq}
Let $\NE$ be a quasi-strong equilibrium of a normal-form game. Then, it is an interior Nash equilibrium of the reduced game at $\NE$.
\end{corollary}

\begin{proof}
    For each player $n \in [N]$ and alternative $i \in \mathrm{supp}(x_n^*)$, the probability mass on $i$ is bounded away from zero since it is positive $x_{n,i}^* > 0$. Then, either it is the unique action that is supported, in which we case we say that $x_n^*$ is trivially contained in the interior of the reduced simplex $\Omega_n(x_n^*)$. Otherwise, it must be bounded away from one $x_{n,i}^* < 1$, since there is at least one other supported action. Either way: 
    \[x_n^* \in \mathrm{int}\big(\Omega_n(x_n^*)\big).\]
    Thus, $\NE$ is an interior equilibrium of the reduced game.
\end{proof}

\subsection{Calculus on the probability simplex}
\label{sec:calculus}
This section formalizes smoothness on the boundary of the simplex. Let $\Delta \equiv \Delta^{k-1}$ be the $(k-1)$-dimensional probability simplex. We view it as the union of open manifolds of dimensions ranging from 0 through $k-1$. In particular, we decompose the simplex as follows:
\[\Delta = \bigcup_{S \subset [k]} \Delta_S, \qquad \textrm{where } \Delta_S := \big\{ x \in \Delta^{k-1} : \mathrm{supp}(x) = S\big\},\]
so that $\Delta_S$ is the interior of a probability simplex with $|S|-1$ dimensions. For each $S \subset [k]$, we say that the closure of $\Delta_S$, denoted by $\overline{\Delta}_S$, is a \emph{face} of the simplex. To take the gradient and Hessian of map $h : \Delta \to \RR$ at a point $x \in \Delta$ with $\mathrm{supp}(x) = S$, we simply take them with respect to the simplex $\Delta_S$.

\begin{definition}[Calculus on the simplex] \label{def:simplex-calculus}
    For each $S \subset [k]$, let $T\Delta_S$ denote the tangent space of $\Delta_S$, which we may naturally identify with the following subspace of $\RR^k$:
    \[T\Delta_S := \big\{v \in \RR^k : \mathrm{supp}(v) = S \textrm{ and } \ind^\top v = 0\big\}.\]
    When a map $h_S : \Delta_S \to \RR$ is smooth, let $\nabla_S h_S : \Delta_S \to T\Delta_S$ denote its gradient. Let $h : \Delta \to \RR$ be a map whose restriction onto $\Delta_S$ is smooth for all $S \subset [k]$. Its \emph{gradient} and \emph{Hessian} at $x$ are defined as:
    \[\nabla h(x) := \nabla_S h(x) \qquad \textrm{and}\qquad \nabla^2 h(x) := \nabla_S^2 h(x),\]
    where $S = \mathrm{supp}(x)$. Thus, $\nabla h(x) \in \RR^k$ and $\nabla^2 h(x) \in \RR^{k \times k}$.
\end{definition}

\begin{definition}[Smoothness on the simplex]
    A map $g : \Delta \to \RR$ is \emph{smooth} on the simplex if it can be extended to a smooth map on an open set $U \subset \RR^k$ containing $\Delta$.
\end{definition}

Any utility $f : \domains \to \RR$ in a normal-form game is smooth on the simplex, since it immediately extends to a multilinear map $p : \RR^{k_1} \times \dotsm \times \RR^{k_N} \to \RR$, where $p$ is formally the same polynomial as $f$. While their function values coincide on $\domains$, they are not the same function because their domains are distinct. In particular, their derivatives are not generally equal, as they map to distinct tangent spaces. Nevertheless, we will use the representation inherited from the ambient space for each $\Delta^{k_n - 1} \subset \RR^{k_n}$ specified in \Cref{def:simplex-calculus}.

\subsection{The rate of suboptimal plays in smoothed-best responses}

\begin{example}[Entropic regularizer is linearly steep]
    \label{ex:entropy-linear-steep}
    Let $h(x) = \sum x_i \log x_i$ on $\Omega = \Delta^{k-1}$. For any $v \in T\Omega$, the solution $x^\beta(v) := \nabla h^{-1}(v/\beta)$ is given by:
    \[x^\beta(v)_i = \frac{\exp(v_i/\beta)}{\sum_{j=1}^{k} \exp(v_j/\beta)} \leq \exp\left(- \frac{1}{\beta} \cdot \Big(\max_{j \in [k]} v_j - v_i\Big)\right).\]
    It follows that the entropic map is linearly steep since $\exp(-\epsilon/\beta)$ decays to zero much faster than $\beta$.
\end{example}

\dominatedrate*

\begin{proof}[Proof of \Cref{lem:steep-convergence}]
    By \Cref{lem:local-domination}, there exists an open region $U$ around $\NE$ and some $\epsilon > 0$ such that for each player $n \in [N]$ and for all $\strats \in U$ and $i \notin \mathrm{supp}(\NE)$, the action $i$ is $\epsilon$-suboptimal. In particular, $\nabla_n f_n(\strats) \in T\Omega_n$ denote the gradient of $f_n$ on the manifold $\Omega_n$. We obtain that:
    \[\nabla_n f_n(\strats)_i < \max_{j \in [k_n]}\, \nabla_n f_n(\strats)_j - \epsilon.\]
    Since $(\SE)_\beta$ converges to $\NE$, for sufficiently small $\beta$, the $\beta$-smoothed equilibrium $\SE$ eventually remains in $U$ as $\beta$ goes to zero. In particular, $\nabla f_n(\SE)$ is contained in the set $V_i(\epsilon)$ where the $i$th alternative is $\epsilon$-suboptimal, defined in \Cref{def:linear-steepness}. As the regularizer $h_n$ is linearly steep, the $\beta$-smoothed equilibrium places sublinearly little mass on any action $i \notin \mathrm{supp}(x_n^*)$: 
    \[\lim_{\beta \to 0} \, \frac{x^\beta_{n,i}}{\beta} = 0.\]
    Thus, the projection of the smoothed equilibrium $\SE$ onto $\domains(\NE)$ satisfies:
    \begin{align*} 
        \frac{\big\|\SE - \Pi \SE \big\|}{\beta} \leq \sum_{n \in [N]} \sum_{i \notin \mathrm{supp}(x_n^*)} \frac{2x^\beta_{n,i}}{\beta},
    \end{align*}
    by the triangle inequality. Taking the limit as $\beta$ goes to zero yields the result.
\end{proof}

\begin{lemma}[Hessians of linearly steep regularizers] \label{lem:choice-regularizer}
    Let $\Delta = \Delta^{k-1}$ be the $(k-1)$-dimensional simplex. Let $H \in T\Delta \otimes T\Delta$ be positive definite and $x \in \mathrm{int}(\Delta)$. There is a linearly steep regularizer $h : \Delta \to \RR$ where: 
    \[\nabla h(x) = 0 \qquad\textrm{and}\qquad \nabla^2 h(x) = H.\]
\end{lemma}

\begin{proof}
    Define $\mathrm{Sym}^+(T\Delta)$ to be the set of positive-definite matrices in $T\Delta \otimes T\Delta$. Fix $x \in \mathrm{int}(\Omega_i)$ and consider the family of steep regularizers:
    \[h(x;\lambda, A, w) = \lambda \sum_{i=1}^{k} x_i \log x_i + \frac{1}{2} \big\|A (x - w)\big\|_2^2,\]
    where $\lambda > 0$, $A \in \RR^{k \times k}$ is invertible, and $w \in \RR^{m_i}$. We show that $h$ is steep and that for this fixed $x$, 
    \[\mathrm{Sym}^+(T\Delta) = \Big\{\nabla^2 h(x;\lambda, A) : \lambda > 0 \textrm{ and } A \textrm{ invertible}\Big\}.\]
    
    Under the coordinate representation given by \Cref{def:simplex-calculus}, the gradient on the simplex $\Delta$ is given by:
    \[\nabla h(x;\lambda, A) = \left(I - \frac{1}{k} \ind \ind^\top\right)\Big[\lambda \log x + A^\top A(x-w)\Big],\]
    where $\log x$ is applied component-wise. Thus, $h$ is steep, since $\log x_i \to - \infty$ as $x_i \to 0$. There exists $w \in \RR^{k}$ such that $\nabla h(x_i) = 0$, by selecting:
    \[w = x_i - \lambda (A^\top A)^{-1} \log x_i.\]
    
    The Hessian on the simplex is given by:
    \[\nabla^2 h(x;\lambda, A) = \left(I - \frac{1}{k} \ind \ind^\top\right)\Big[\lambda \cdot \mathrm{diag}(1/x)  + A^\top A\Big]\left(I - \frac{1}{k} \ind \ind^\top\right),\]
    where $\mathrm{diag}(1/x)$ is the diagonal matrix with diagonal entries $1/x_i$. For every $M \in \mathrm{Sym}^+(T\Delta)$, there is a sufficiently small $\lambda > 0$ such that the following matrix is positive-definite, so there is an invertible matrix $A$:
    \[A^\top A = M + \ind \ind^\top - \lambda \cdot \mathrm{diag}(1/x).\]
    It follows that $\nabla^2 h(x;\lambda, A) = M$, showing that $\mathrm{Sym}^+(T\Delta) \subset \{\nabla_{\Omega_i}^2 h_i(z;\lambda, A): \lambda > 0 \textrm{ and } A \textrm{ invertible}\}$. The reverse inclusion holds as $h$ is strongly convex, so its Hessian on the simplex is contained in $\mathrm{Sym}^+(T\Delta)$.
\end{proof}

\subsection{Convergence to non-interior Nash equilibria}

\boundaryconv*
\begin{proof}
    Let $\betaBRs$ be any choice of smoothed best-response map induced by a set of linearly steep and proper regularizers $\regs$ and $\beta > 0$. Its Jacobian under the coordinate representation specified by \Cref{def:simplex-calculus} is:
    \[\nabla \betaBRs(\strats) = \frac{1}{\beta} \Hs(\strats)^{+} \Jacobian(\strats),\]
    where $\Hs(\strats)$ is the block-diagonal matrix with the $n$th block being $\nabla^2_n h_n(x_n)^+$ and $\Jacobian(\strats)$ is the game Jacobian, as computed by \Cref{lem:gradient-sbr}. Define $\mathbf{G}(\strats) := \Hs(\strats)^+ \Jacobian(\strats)$ and let $\Pi : \domains \to \domains(\NE)$ be the orthogonal projection to the support of $\NE$. Let $U \subset \domains$ be any open ball containing $\NE$ such that $\|\Hs(\strats)^{+} \Jacobian(\strats)\| \leq L$ is bounded by a constant $L > 0$. Such a ball exists because the regularizers are proper, so that $\nabla \betaBRs$ is smooth. Moreover, the reduced game around $\NE$ is locally uniformly stable, so this ball can be chosen such that:
    \[\mathrm{spec}\big(\mathbf{G}(\Pi \SE)\big) \subset i\RR.\]
    
    To prove that the $(\betaBRs, \eta)$-averaging dynamics is a local contraction at $\SE$ whenever $\eta$ is sufficiently small, we show that the Jacobian of these dynamics have operator norm bounded by 1. The Jacobian of the $(\betaBRs, \eta)$-averaging dynamics can be decomposed as follows:
    \begin{align*} 
        (1 - \eta) \mathbf{I} + \eta \nabla \betaBRs(\SE) &= (1 - \eta)\mathbf{I} + \frac{\eta}{\beta}\, \mathbf{G}(\SE)
        \\&= (1 - \eta)\mathbf{I} + \frac{\eta}{\beta}\, \mathbf{G}(\Pi \SE) + \frac{\eta}{\beta} \,\bigg(\mathbf{G}(\SE) - \mathbf{G}(\Pi \SE)\bigg)
        \\&= (1 - \eta)\mathbf{I} + \frac{\eta}{\beta}\, \mathbf{G}(\Pi \SE) + \eta\,\bigg(\frac{\nabla \mathbf{G}(\bar{\strats}^\beta)(\SE - \Pi \SE)}{\beta}\bigg),
    \end{align*}
    where in the last step applies the mean-value theorem to find some $\bar{\strats}^\beta$ that is a convex combination of $\SE$ and $\Pi \SE$, once again making use of the smoothness of the pseudoinverse of the Hessian of the regularizers.

    When $\beta$ is sufficiently small, then $\Pi \SE$ is contained in $U$, so that the eigenvalue of $\mathbf{G}(\Pi \SE)$ is $i\lambda$ for some real-valued scalar $\lambda$ with $|\lambda| < L$. By \Cref{lem:steep-convergence}, the norm of the last term can be made arbitrarily small:
    \begin{align*}
        \left\|\frac{\nabla \mathbf{G}(\bar{\strats}^\beta)(\SE - \Pi \SE)}{\beta}\right\|_2 &\leq \frac{C \|\SE - \Pi \SE\|}{\beta} \leq o(\beta),
    \end{align*}
    where we may uniformly bound the norm of $\nabla \mathbf{G}$ over $\domains$ since it is a smooth map on a compact set. Let us choose $\beta$ small enough so that this term is no more than $1/2$. Thus, the operator norm of the Jacobian of the dynamics is bounded by:
    \begin{align*}
        \big\|(1 - \eta) \mathbf{I} + \eta \nabla \betaBRs(\SE)\big\|_2 &\leq \left|1 - \frac{\eta}{2} + i \eta \frac{L}{\beta}\right| \leq \exp\left(- \frac{\eta}{2}\right),
    \end{align*}
    where the last inequality holds when $\eta \leq \beta^2 / (1 + 4L^2)$; see the proof of \Cref{thm:main-convergence} for this computation.
\end{proof}

\end{document}

%% file: preamble.tex
\usepackage{xcolor}
\definecolor{color1}{HTML}{105e8a}
\definecolor{color2}{HTML}{e99926}
\definecolor{color3}{HTML}{b82a0c}
\definecolor{color4}{HTML}{3e8a10}
\definecolor{color5}{HTML}{80037e}
\definecolor{color6}{HTML}{070707}

\definecolor{gradient0}{RGB}{190, 245, 250}
\definecolor{gradient1}{RGB}{164, 233, 248}
\definecolor{gradient2}{RGB}{134, 221, 248}
\definecolor{gradient3}{RGB}{104, 208, 249}
\definecolor{gradient4}{RGB}{72, 195, 252}
\definecolor{gradient5}{RGB}{36, 181, 255}

\usepackage{hyperref}
\hypersetup{
    colorlinks,
    linkcolor={color1},
    citecolor={color6},
    urlcolor={color2}
}

\usepackage{amsmath,amssymb,amsthm}
\usepackage{bbm,dsfont}
\usepackage{mathtools}
\usepackage{apptools}

\usepackage{thmtools,thm-restate}
\newtheorem{theorem}{Theorem}%[section]

\newtheorem{lemma}{Lemma}

\newtheorem{example}{Example}

\newtheorem{corollary}{Corollary}

\theoremstyle{definition}
\newtheorem{definition}{Definition}

\newtheorem*{informaltheorem*}{Theorem}

\AtAppendix{\counterwithin{theorem}{section}}
\AtAppendix{\counterwithin{lemma}{section}}
\AtAppendix{\counterwithin{proposition}{section}}
\AtAppendix{\counterwithin{definition}{section}}
\AtAppendix{\counterwithin{example}{section}}
\AtAppendix{\counterwithin{corollary}{section}}

\usepackage{caption}
\usepackage{booktabs}
\usepackage{nicefrac}      
\usepackage{afterpage}
\usepackage{inconsolata}
\usepackage{units}
\usepackage{fancyvrb}
\fvset{fontsize=\normalsize}
\usepackage{multicol}
\usepackage{comment}
\usepackage{enumitem}

\usepackage{multirow,array}
\usepackage{placeins}
\usepackage{pgfplots}
\pgfplotsset{compat=1.17}
\usepackage{tikz}
\usetikzlibrary{calc}
\usetikzlibrary{patterns,patterns.meta}
\usetikzlibrary{positioning,cd}
\usepackage{transparent}

\usepackage{flafter}
\usepackage{comment}
\usepackage{cleveref}

\usepackage{mdframed}
\newmdenv[
  topline=false,
  bottomline=false,
  rightline=false,
  skipabove=\topsep,
  skipbelow=\topsep,
  innertopmargin=0pt,
  innerbottommargin=0pt
]{siderules}

\usepackage{ragged2e}

\newcommand{\cL}{\mathcal{L}}

\newcommand{\CC}{\mathbb{C}} % Complex numbers
\newcommand{\NN}{\mathbb{N}} % Natural numbers
\newcommand{\RR}{\mathbb{R}} % Real numbers
\renewcommand{\epsilon}{\varepsilon}

\DeclareMathOperator*{\argmax}{\mathrm{arg\,max}}

\newcommand{\ind}{\mathbbm{1}}

\usepackage{mathtools}

\counterwithout{figure}{section}
\counterwithout{equation}{section}

\def\ind{\mathbbm{1}}

\usepackage{algorithm}
\usepackage{algorithmic}

\usepackage{natbib}
\bibpunct{(}{)}{;}{a}{,}{,}

%% file: doc/1-intro.tex
The \cite{nash1951non} equilibrium is a foundational solution concept in games, capturing when collective behavior or strategies may be stationary. These equilibria are meaningful to study, for once such strategies appear, they may persist for a long time. But, there is an important caveat: not all equilibria can be robustly reached by players in the game \citep{hart2003uncoupled,papadimitriou2007complexity,daskalakis2009complexity,daskalakis2010learning,milionis2023impossibility}. Any equilibrium that cannot be found or sustained is unlikely to have practical relevance. This caveat leads to the question on \emph{learnability}: which Nash equilibria can players eventually learn to play from repeated interactions?

To grasp individual and corporate behaviors, literature in classical economics considers a model of learning where players: (i) take an evolutionary approach and incrementally update their strategies; (ii) are utility-seeking, which means that they aim to improve their own payoffs; and (iii) are uncoupled, where players are unaware of the other players' utilities or methods to improve them~\citep{alchian1950uncertainty,winter1971satisficing}. The goal of the players in this model is not to compute any pre-determined Nash equilibrium, at least not explicitly. Rather, the equilibrium is to emerge out of their joint, but individually  utility-seeking, behavior.

Within this model of learning, the impossibility result of \cite{hart2003uncoupled} shows that no \emph{uncoupled} and \emph{asymptotically-stable} learning dynamics can converge to all Nash equilibria. Uncoupled means that players learn in uncoordinated and decentralized ways. Asymptotic stability means that the convergence is robust to small perturbations in the dynamics. Simply put, not all equilibria admit such a strong notion of learnability, viz. dynamical stability. The ones that do, however, have been characterized for certain classes of learning dynamics in standard normal-form games: an equilibrium is `asymptotically learnable' in this way if and only if it is  \emph{strict}, where every player has a single, deterministic strategy that is clearly locally optimal \citep{samuelson1992evolutionary,vlatakis2020no,giannou2021survival}. 

A significant gap remains for the learnability of \emph{mixed} Nash equilibria, where players may use randomized strategies.
On the one hand, mixed equilibria are not strict, and as a result, they are not asymptotically stable under these learning dynamics. Observations corroborating this finding demonstrate that many dynamics are not able to generically learn mixed Nash equilibria. This has been a significant source of criticism on the viability of mixed equilibria as solutions in games \citep{shapley1963some,crawford1985learning,stahl1988instability,jordan1993three,krishna1998convergence,ellison2000learning,kleinberg2011beyond,mertikopoulos2018cycles,bailey2021stochastic}. On the other hand, there are cases in which mixed equilibria \emph{are} approachable by simple, uncoupled dynamics, such as those in two-player zero-sum games \citep{robinson1951iterative,fudenberg1993learning,gjerstad1996rate,hofbauer1998evolutionary,benaim1999mixed,hofbauer2006best,daskalakis2018last,wei2021linear,piliouras2022fast}. 
To address this conundrum, we must go beyond the criterion of asymptotic stability, which is too stringent a requirement.

This work extends the study of learning in games to dynamics that exhibit \emph{non-asymptotic} stability. It is a weaker notion of stability that includes neutral stability, which only requires that players starting off in a neighborhood of such an equilibrium will continue to remain close to it. The question we ask is:
\begin{center}
\begin{minipage}{0.8\textwidth}
\centering
  \textit{Under the relaxed criterion of non-asymptotic stability, which Nash equilibria are learnable by uncoupled dynamics, and what are their economic properties?} 
\end{minipage}
\end{center}

There are two parts to this work. In the first part, we introduce and characterize a form of non-asymptotic stability that applies to broad classes of learning dynamics. We call it \emph{uniform stability}. While it arises from dynamical considerations, we show that uniform stability is closely tied to the economic properties of the equilibrium, namely, \emph{strategic Pareto optimality}. This is a game-variant of weak Pareto optimality defining a minimal form of collective rationality---it describes solutions where there is no way to strictly improve everyone's utilities, up to strategic equivalence.\footnote{As standard learning dynamics generally factor through strategic equivalence classes of games (they do not depend on the `externalities' or non-strategic components of the game), this reservation is necessary.} We show that an equilibrium that is not uniformly stable must  not be strategically Pareto optimal. This formalizes and tightens an observation made by prior work: non-convergence to mixed equilibria can sometimes actually be a blessing in disguise, where players would be worse off had they managed to converge to the equilibrium \citep{kleinberg2011beyond,ostrovski2013payoff,pangallo2022towards,anagnostides2022last}. This may be surprising as it runs counter to the more prominent phenomenon where non-cooperative, individual behavior leads to worse social outcomes, or a high \emph{price of anarchy}, such as in the Prisoner's Dilemma or the Tragedy of the Commons \citep{luce1957games,hardin1968tragedy,nachbar1990evolutionary,weibull1994if,koutsoupias1999worst}.

In the second part, we focus on convergence and non-convergence of \emph{incremental smoothed best-response dynamics}, a generalization of a canonical family of learning dynamics. For this class, we show that non-strict, uniformly-stable equilibria can be learned, though with slower convergence than strict equilibria. Around a locally uniformly-stable Nash equilibrium, these dynamics can be \emph{stabilized}: they lead to asymptotically-stable dynamics that approximate the equilibrium to arbitrary accuracy. While for non-uniformly stable mixed equilibria, there are dynamics that can never be stabilized to the equilibrium. 

Thus we observe a dichotomy between dynamics around strict and mixed Nash equilibria under the context of `as if' rationality~\citep{friedman1953methodology,weibull1994if}.
In the former, individually utility-seeking dynamics lead to the same equilibrium behaviors as if all the players have absolute, individual rationality.
As in the prisoner's dilemma, they can stabilize at collectively-irrational behaviors. 
Around mixed equilibria, uniform learnability disallows collective irrationality.
As a result, individually utility-seeking dynamics robustly behave as if the players were collectively rational, reminiscent of the `invisible hand' from classical economic theory~\citep{smith1776inquiry,debreu1959theory}.

\subsection{Main results}

We consider uncoupled learning dynamics for standard $N$-player normal-form games, where players, through repeated interactions, evolve their individual strategies in order to maximize their own utilities. 

For such dynamics, in \Cref{subsec:equilibrium-classes}, we introduce two non-asymptotic stability classes for mixed equilibria: those that have \emph{pointwise uniform stability} and those that exhibit a stronger \emph{local uniform stability}. These impose a second-order differential condition, either only at the equilibrium or in a neighborhood around it (note that the equilibrium condition itself is a first-order condition for mixed strategies). These conditions naturally arise from linear stability analysis, and they apply to broad classes of dynamics such as smoothed best-response, gradient ascent/adaptive dynamics, mirror ascent/positive-definite adaptive dynamics. 

\begin{restatable}{theorem}{localmeaning}
    \label{thm:meaning-local}
    Consider any general-sum $N$-player normal-form game with a mixed Nash equilibrium. Suppose that the equilibrium is locally uniformly stable. Then, it is locally strategically Pareto optimal.
\end{restatable}

\begin{comment}
In \Cref{sec:economic-meaning}, we show the chain of implications for mixed Nash equilibria under mild assumptions:
\[\textrm{local uniform stability} \quad \implies\quad \textrm{strategic Pareto optimality}.
% \quad \implies\quad \textrm{pointwise uniform stability}.
\]
Moreover, all of these solution concepts are equivalent for the class of \emph{graphical} or \emph{polymatrix games}, which are particularly succinct normal-form games that decompose across pairwise interactions between players.\end{comment}

In \Cref{sec:convergence}, we prove convergence results for the family of \emph{incremental, smoothed best-response learning dynamics}. The dynamics are controlled by a smoothing parameter $\beta$ and a learning rate $\eta$. The first governs the quality of any approximate Nash equilibria that the dynamics finds, while the second impacts the stability of its fixed points; whether the dynamics converge and at what rate.
We prove that if a Nash equilibrium is locally uniformly stable, then for any choice of approximation $\beta$, the dynamics can always be stabilized to the equilibrium by taking a sufficiently fine learning rate $\eta$, whereby it converges with asymptotic stability. With $T$ iterations, the overall convergence rate to the locally uniformly stable mixed Nash equilibria is of order $T^{-1/2}$.
However, if a Nash equilibrium is not pointwise uniformly stable, then there are dynamics that can never be stabilized beyond a certain approximation error, no matter what learning rate is chosen. 
\[\textrm{local uniform stability} \quad \implies\quad \textrm{convergence to mixed Nash equilibrium}
\quad \implies\quad \textrm{pointwise uniform stability}.
\]

Finally, as the first set of results are stated for fully-mixed equilibria, \Cref{sec:boundary} extends these results to the case where the Nash equilibrium may be partially mixed. The following summarizes our main results:

\begin{center}
\begin{minipage}{0.95\textwidth}
    \textbf{Local uniform stability implies strategic Pareto optimality.} If a Nash equilibrium falls within a region of uniformly-stable strategies, it must be collectively rational (\Cref{thm:meaning-local}). 
    % In this case, it is robustly approximable under incremental, smoothed best-response dynamics (\Cref{thm:main-convergence,thm:boundary-convergence}).
\end{minipage}
\end{center}

% \begin{center}
% \begin{minipage}{0.95\textwidth}
%     \textbf{Strategic Pareto optimality implies pointwise uniform stability.} We show that strategic Pareto optimality implies \emph{strategic Pareto stationarity}, which is a second-order condition for collective rationality (\Cref{prop:necessary}). Strategic stationarity is equivalent to pointwise uniform stability (\Cref{thm:meaning}), and it is also necessary for convergence with non-asymptotic stability (\Cref{prop:non-convergence-general}).
% \end{minipage}
% \end{center}
\begin{center}
\begin{minipage}{0.95\textwidth}
\textbf{Local uniform stability implies convergence to the mixed Nash equilibrium.} If a mixed Nash equilibrium is locally uniformly stable, then any strategy initialized close enough to the equilibrium will converge to the mixed Nash equilibrium at a $O(T^{-1/2})$ rate (\Cref{thm:main-convergence,thm:boundary-convergence}).
Conversely, pointwise uniform stability is necessary for convergence to the mixed Nash equilibrium (\Cref{prop:non-convergence-general}).
\end{minipage}
\end{center}

\subsection{Related work}
This work considers the convergence of learning dynamics in games. Two modes of convergence are often studied (i) time-averaged convergence, where players aim to achieve little or no regret in the long run, or the stronger (ii) day-to-day or last-iterate convergence, where players eventually learn to play an equilibrium of the game; see also \cite{fudenberg1993learning} for discussion. This work falls in the latter category.

The learnability of strict Nash equilibria is fairly well-understood \citep[and related works therein] {hart2003uncoupled,mertikopoulos2016learning,vlatakis2020no,giannou2021survival}. In contrast, the learnability of mixed Nash equilibria is more complicated. Many works have pointed out that asymptotically-stable convergence to mixed Nash equilibria is generally ruled out because linearized dynamics about them have zero trace \citep{crawford1985learning,jordan1993three,singh2000nash,hart2003uncoupled,mertikopoulos2018cycles}. As mixed Nash equilibria are needed for the existence of Nash equilibria \citep{nash1951non}, their generic infeasibility is troubling for the viability of the solution concept, but has a resolution via Harsanyi's purification theorem \citep{harsanyi1973games}. 

Philosophically, unlike \cite{harsanyi1973games}, this work does not try to save all mixed equilibria, but proceeds from the fact that uncoupled dynamics can sometimes converge to mixed Nash equilibria, and aims to further refine the solution concept \citep{van1992refinements}. Much of the existing work providing positive convergence results focuses on two-player games \citep{robinson1951iterative,fudenberg1993learning,gjerstad1996rate,hofbauer1998evolutionary,benaim1999mixed,hofbauer2006best,wei2021linear}, and especially for zero-sum games \citep[and related works therein]{gilpin2012first,daskalakis2018last,de2022optimistic,piliouras2022fast}. Our results apply to general $N$-player normal-form games.

The results and techniques in this paper also relate to many areas of game theory. The connection between uniform stability and weak Pareto optimality generalizes existing characterizations of games with weakly Pareto optimal equilibria in two-player games \citep{moulin1978strategically,adler2009note}, and bridges to the classical notion of strong Nash equilibria \citep{aumann2016acceptable}. Our notion of uniform stability builds on the game Jacobian (also sometimes called the game Hessian), which is also generally used to analyze smooth games \citep{isaacs1999differential,candogan2011flows,balduzzi2018mechanics,letcher2019differentiable}.

The algorithmic component of this paper focuses on smoothed best-response dynamics, which can be viewed as regularized learning/optimization \citep{rockafellar1997convex,mertikopoulos2016learning}. In normal-form games, the canonical regularizer is the logit map, yielding quantal response equilibria \citep{mckelvey1995quantal,alos2010logit,goeree2016quantal,mertikopoulos2016learning}. Our convergence result introduces the notion of stabilization to an equilibrium, which is similar to the tracing procedure discussed by \cite{mckelvey1995quantal}.

%% file: doc/2-preliminaries.tex
An \emph{$N$-player game} $(\domains, \utilities)$ with players indexed by $n \in [N]$ consists of a \emph{joint strategy space} $\domains = \Omega_1 \times \dotsm \times \Omega_N$ and a set of \emph{utilities} $\utilities = (f_1,\ldots, f_N)$, where each player would like to maximize their own utility $f_n : \domains \to \RR$.
In a round of the game, everyone independently chooses a strategy $x_n \in \Omega_n$. These choices constitute the joint decision vector $\strats = (x_1,\ldots, x_N)$, and players receive their respective payoffs $f_n(\strats) \equiv f_n(x_n;\strats_{-n})$. Here, we let $\strats_{-n} = (x_1,\ldots, x_{n-1}, x_{n+1}, \ldots, x_N)$ denote the joint strategy without the $n$th component.

This work focuses on \emph{normal-form games}, which we define here and expand upon in \Cref{sec:normal-form-games-exposition}.

\begin{definition}[Normal-form game]
    An $N$-player \emph{normal-form game} $(\domains, \utilities)$ is one where each player $n \in [N]$ has a set of $k_n$ alternatives/pure strategies, the strategy space $\Omega_n$ is the $(k_n-1)$-dimensional probability simplex $\Delta^{k_n - 1}$, and each utility $f_n$ is defined by a payoff tensor $T^n \in \RR^{k_1 \times \dotsm \times k_N}$ as follows:
    \[f_n(\strats) = \sum_{i_1 \in [k_1]}\dotsm\sum_{i_N \in [k_N]} T^n_{i_1,\ldots, i_N} x_{1,i_1} \dotsm x_{N,i_N}\quad \textrm{and}\quad x_n = (x_{n,1},\ldots, x_{n,k_n}) \in \Delta^{k_n - 1}.\]
    We say that the strategy $\strats \in \domains$ is \emph{mixed} if it is in the interior of $\domains$, which we denote by $\mathrm{int}(\domains)$.
\end{definition}

Normal-form games are a special case of \emph{multilinear games}, which are games whose utilities are polynomials where the maximum degree in each variable is one. In the following, we first define multilinear polynomials of the form $p : \RR^{d_1} \times \dotsm \times \RR^{d_N} \to \RR$, where we think of $d_n = k_n -1$. The connection between normal-form games and multilinear games is formalized by \Cref{lem:normal-to-multilinear}.

\begin{definition}[Multilinear polynomial] \label{def:multilinear-polynomial}
    Let $\RR[\strats]$ contain the polynomials of the form $p : \RR^{d_1} \times \dotsm \times \RR^{d_N} \to \RR$, with variables $\strats = (x_1,\ldots, x_N)$ and $x_n = (x_{n,1},\ldots, x_{n,d_n})$ for each $n \in [N]$. We say that $p \in \RR[\strats]$ is \emph{multilinear} over $\strats$ if it is affine over each $x_n$. That is, for each $n \in [N]$, the polynomial is affine in $x_n$,
    \[p(\strats) = A_{n}(\strats_{-n})^\top x_n + b_{n}(\strats_{-n}),\]
    where $A_n \in \RR[\strats_{-n}]^{d_n}$ and $b_n \in \RR[\strats_{-n}]$ are polynomials, and $\strats_{-n}$ includes all variables besides those in $x_n$.
\end{definition}

\begin{definition}[Multilinear game]
    An $N$-player \emph{multilinear game} $(\domains,\utilities)$ is a game where for each $n \in [N]$, the strategy space $\Omega_n$ is an open, convex subset of $\RR^{d_n}$, and the utility $f_n : \domains \to \RR$ is a multilinear polynomial. That is, each utility $f_n$ extends to a multilinear polynomial on all of $\RR^{d_1} \times \dotsm \times \RR^{d_N}$.
\end{definition}

\subsection{Individual and collective rationality: two solution concepts}
A fundamental aspect of non-cooperative games is that, even though players cannot directly control how others behave, the payoff that each receives nevertheless depends on their joint decision. This leads to the Nash equilibrium as a solution concept for games---joint strategies where it is not possible to improve one's own utility without coordinating with other players. It is a notion of \emph{individual rationality}:

\begin{definition}[Nash equilibrium]
    A joint decision vector $\NE \in \domains$ is a \emph{Nash equilibrium} if:
    \[\forall n \in [N],\qquad x_n^* \in \argmax_{x_n \in \Omega_n}\, f_n(x_n; \NE_{-n}).\hspace{25pt}\]
\end{definition}

In contrast, if as in the multi-objective optimization, players are cooperative and can coordinate, then a minimal solution concept is weak Pareto optimality, which describes joint decisions where there is no way to strictly improve everyone's utilities. It is a notion of \emph{collective rationality}:

\begin{definition}[Weak Pareto optimality]
    A joint decision vector $\NE \in \domains$ is \emph{weakly Pareto optimal} if there does not exist any $\strats \in \domains$ such that $f_n(\strats) > f_n(\NE)$ for all $n \in [N]$.
\end{definition}

These solution concepts are generally incomparable and independent of each other. Famous examples, including the \emph{prisoner's dilemma} or the \emph{tragedy of the commons}, demonstrate that Nash equilibria and Pareto optimal strategies can even be disjoint---in the absence of mechanisms to enable or enforce cooperation, socially-optimal outcomes can be unstable under individually-rational or selfish behavior.

This work refines the relationship between Nash stability and Pareto optimality from the perspective of \emph{learning in games}, which is a dynamical perspective in game theory motivated by a simple observation: a solution concept has no bearing on practice if players of a game cannot find it. Eventually, we will show that the mixed Nash equilibria that can be robustly found must, in fact, be strategically Pareto optimal, a notion related to Pareto optimality that is well-suited for games. We will motivate and introduce this in the next section. To do so, we now turn to how players come to discover and play a Nash equilibrium.

\subsection{Uncoupled learning dynamics and strategic equivalence}

Consider a model of learning where players learn to play the game via trial-and-error: they simply play many rounds of the game. As players may observe the joint strategies $\strats(t)$ chosen over time $t$, they can also continually adjust their own strategies. A \emph{learning rule/dynamics} specifies how a player $n \in [N]$ chooses to act each round $x_n(t)$ as a function of past experiences, namely $\strats(s)$ where $s < t$, as well as any additional information they may have about the game---for example, we assume players know their own utilities.

A basic informational constraint we impose is that the players are \emph{uncoupled}, which means that a player's learning rule cannot depend on the utilities or learning rules of other players \citep{hart2003uncoupled}. Intuitively, this means that players do not know what others want, nor how others learn from past experiences. Even if they were so inclined, players do not know how to directly be altruistic or malicious. This constraint in a sense enforces selfish behavior, or individual rationality. 

To discuss the connection between the economic and the dynamic properties of games, we need to focus on the `active ingredients' in the utilities that play a role in the dynamics. In existing literature, this motivates the decomposition of utilities into their \emph{strategic} and \emph{non-strategic} components \citep{candogan2011flows}. In multilinear games, the strategic component precisely consists of the terms in the polynomial $f_n$ that depends on the $n$th set of variables $x_n$, while the non-strategic terms are those that vary independently of $x_n$. Games with the same strategic components are called \emph{strategically equivalent} \citep{monderer1996potential}.

\begin{comment}
\begin{definition}[Strategic and non-strategic components]
    \label{def:strategic-component}
    Let $(\domains, \utilities)$ be a multilinear game. For each $n \in [N]$, the utility $f_n$ is a polynomial that is affine in $x_n$, with the following decomposition given in \Cref{def:multilinear-polynomial},
    \[f_n(\strats) = \underbrace{A_n(\strats_{-n})^\top  x_n}_{\textrm{strategic}} \,\,\,+\, \underbrace{b_n(\strats_{-n})}_{\textrm{non-strategic}}\!\!.\] 
    The \emph{strategic component} of $f_n$ is the term linear in $x_n$. The \emph{non-strategic} component of $f_n$ is the offset term. Two utilities are \emph{strategically equivalent} if they have the same strategic component, denoted by $f_n \simeq_{\mathrm{SE}} f_n'$.
\end{definition}
\end{comment}

\begin{definition}[Strategic and non-strategic components]
    \label{def:strategic-component}
    Two utilities $f_n, f_n'$ for agent $n$ are \emph{strategically equivalent} if their difference $f_n(\strats) - f_n'(\strats)$ is constant in $x_n$. Given a reference point $\NE$, the \emph{strategic component} of $f_n$ at $\NE$ is the map $\tilde{f}_n$, where
    \[\tilde{f}_n(\strats) = f_n(\strats) - f_n(x_n^*; \strats_{-n}).\]
    In particular, two utilities are strategically equivalent if and only if they have the same strategic components.
\end{definition}

The target of this work is to illuminate the game-theoretic properties of equilibria found or not found by simple, uncoupled learning dynamics. As the dynamics we study, \Cref{eqn:dynamics} below, indeed turn out to be invariant to changes in the non-strategic component of the game (\Cref{lem:strategic-equivalence}), we expect results only up to strategic equivalence. This motivates the following notion of strategic Pareto optimality.

\begin{comment}
\begin{definition}[Strategic Pareto optimality]
    \label{def:strategic-pareto-optimality}
    Let $(\domains, \utilities)$ be a multilinear game. A joint decision vector $\NE \in \domains$ is \emph{strategically Pareto optimal} if it is weakly Pareto optimal for the strategic component of the game.
\end{definition}
\end{comment}

\begin{definition}[Strategic Pareto optimality]
    \label{def:strategic-pareto-optimality}
    Let $\NE \in \domains$. It is  \emph{strategically Pareto optimal} if it is weakly Pareto optimal for the strategic components $(\tilde{f}_1,\ldots, \tilde{f}_N)$ at $\NE$.
\end{definition}

Akin to the classic story about Nash equilibria and weak Pareto optimality, a simple example shows that neither the condition of being a Nash equilibrium nor being strategically Pareto optimal imply each other (\Cref{ex:strategic-not-nash}). 
% But surprisingly, despite uncoupledness, players can be guaranteed to converge to a mixed Nash equilibrium that is strategically Pareto optimal, while they are not guaranteed to converge to a mixed Nash equilibrium without strategic Pareto optimality under mild non-degeneracy assumptions. 
But surprisingly, despite uncoupledness, players can only  converge to a mixed Nash equilibrium that are strategically Pareto optimal. 
In the next section, we formally introduce the learning dynamics considered in this work.

\subsection{A simple class of learning rules}

Let $(\domains, \utilities)$ be an $N$-player game where each $\Omega_n$ is a convex domain. Let $\regs = (h_1,\ldots, h_N)$ be a collection of strictly convex regularizers $h_n : \Omega_n \to \RR$. The dynamics we study are based on smoothed best-responses, where players maximize their own utilities by solving a regularized optimization problem:

\begin{definition}[$\beta$-smoothed best-response] \label{def:smoothed-br}
    The \emph{$\beta$-smoothed best-response map} $\betaBRs = (\betaBR_1,\ldots, \betaBR_N) : \domains \to \domains$ with respect to a set of strictly convex regularizers $\regs$ for any $\beta > 0$ is given by:
    \begin{equation*} %\label{eqn:smoothed-br}
        \forall n \in [N],\qquad \betaBR_n(\strats) = \argmax_{x_n' \in \Omega_n}\, f_n^{\vphantom{\prime}}(x_n'; \strats_{-n}^{\vphantom{\prime}}) - \beta h_n^{\vphantom{\prime}}(x_n'). \phantom{\forall n \in [N]}
    \end{equation*}
\end{definition}

Throughout this work, we assume that the regularizers are \emph{steep}, which is a standard assumption from convex optimization that guarantees that $\betaBRs$ map into the interior of $\domains$ \citep[Theorem 26.1]{rockafellar1997convex}.
\begin{definition}[Steep regularizer]
    \label{def:steep-reg}
    Let $\Omega$ be a compact and convex domain, and $h: \Omega \to \RR$ be strictly convex and smooth on $\mathrm{int}(\Omega)$. It is steep if for any sequence $x_k \in \Omega$ converging to a boundary point of $\Omega$,
    \[\lim_{k \to \infty} \, \big\|\nabla h(x_k)\big\| = \infty.\]
\end{definition}
In the case of normal-form games, a canonical choice is entropy regularization, $h_n(x_n) = \sum_i x_{n,i} \log x_{n,i}$, which is a steep regularizer on the probability simplex. The smoothed best-response map is also called \emph{quantal response}, and the smoothed-best responses are given by the softmax function:
\[\Phi_n^\beta(\strats) \propto \exp\Big(\frac{1}{\beta} \nabla_n f_n(\strats)\Big).\]
The specific class of discrete-time dynamics that we study in this work is the following:

\begin{graybox}
\paragraph{Incremental smoothed best-response dynamics.}
Let $\eta \in (0,1)$ be a learning rate and $\betaBRs$ be a $\beta$-smoothed best-response map induced by steep regularizers. The \emph{$(\betaBRs, \eta)$-averaging dynamics} is the dynamics $\strats(t)$ given by the transition map:
\begin{equation} \label{eqn:dynamics}
    \strats(t) = (1 - \eta) \, \strats(t-1) + \eta \, \betaBRs\big(\strats(t-1)\big),
\end{equation}
where the joint strategy $\strats(t)$ played in the round $t$ is an weighted average of the strategy $\strats(t-1)$ played in the previous round $t-1$ with its smoothed best response $\betaBRs(\strats(t-1))$. 
\end{graybox}

The fixed points of \eqref{eqn:dynamics} are precisely the fixed points of the smoothed best-response map $\betaBRs$, which are called \emph{smoothed equilibria}. Importantly, smoothed equilibria are \emph{approximate Nash equilibria} (\Cref{lem:smooth-to-Nash}), and they converge to Nash equilibria as the smoothness parameter $\beta$ goes to zero (\Cref{lem:convergence-to-Nash}). When the dynamics converge, they achieve meaningful solutions in a game-theoretic sense.

\begin{definition}[$\beta$-smoothed equilibrium]
    Let $\betaBRs$ be a $\beta$-smoothed best-response map. A joint decision vector $\SE \in \Omega$ is a \emph{$\beta$-smoothed equilibrium} if it is a fixed point $\SE = \betaBRs(\SE)$.
\end{definition}

\begin{definition}[$\epsilon$-approximate Nash equilibrium]
    Let $\epsilon \geq 0$. We say that a joint decision vector $\strats^\epsilon \in \domains$ is an \emph{$\epsilon$-approximate Nash equilibrium} if $f_n(x_n, \strats_{-n}^\epsilon) \leq f_n(\strats^\epsilon) + \epsilon$ for all $\strats \in \domains$ and $n \in [N]$.
\end{definition}

\paragraph{Economic interpretation of game dynamics} Smoothed best-response is a common model of \emph{bounded rationality}, where the solution that a player finds is only an approximate best-reply to the strategies $\strats_{-n}$ chosen by the others, perhaps due to noise, or constraints on information or computation. The player approaches `perfect rationality' as the parameter $\beta$ goes to zero.

Typically, smoothed best-response and its limiting best-response are rational from the perspective of players that are \emph{myopic}, which means that they do not try to steer the long-term behavior of the dynamics. Myopia is often justified in the evolutionary game theory setting, where each `player' corresponds to a large population of individuals. As each individual has an infinitesimal impact on the long-term behavior of the dynamics, the most individually-rational behavior is to optimize for immediate outcomes. Myopia is also justified if the individual's lifespan is much shorter than the timescale over which the dynamics evolve.

Under this context, the incremental or averaging aspect of the dynamics \eqref{eqn:dynamics} also arise naturally; \cite{fudenberg1998theory} also call these ``partial'' dynamics. It capture situations where only a random fraction of individual within the populations play the game at each round. As only a fraction of individuals update their strategies, the strategy profile only partially evolves toward the smoothed best-response direction. In short, the dynamics \eqref{eqn:dynamics} can be considered as a model for strategic interactions across large, imperfectly-rational populations, such as at a traffic stop where individuals with bounded rationality are intermittently matched up with others to play a dangerous game of chicken.

%% file: doc/3-uniform-stability.tex
The local behavior of a smooth, non-degenerate dynamical system is determined by linearizing the dynamics. In learning in games literature, this motivates the \emph{game Jacobian} (e.g.\ \cite{letcher2019differentiable}). It arises naturally from the analysis of gradient ascent dynamics, where players update their strategies along the direction of steepest ascent $\nabla_n f_n(\strats)$, and the game Jacobian is the Jacobian of the gradient ascent dynamics:

\begin{definition}[Jacobian of the game] \label{def:game-jacobian}
    Let $(\domains, \utilities)$ be an $N$-player multilinear game and let $\strats \in \domains$. The \emph{game Jacobian} $\Jacobian(\strats)$ is the Jacobian of the map $\strats \mapsto (\nabla_1 f_1(\strats), \ldots, \nabla_N f_N(\strats))$, so that:
    \[\forall n,m \in [N],\qquad J_{nm}(\strats) = \nabla^2_{nm} f_n(\strats).\hspace{2em}\]
    In particular, the block diagonals $J_{nn}(\strats) = 0$ are zero matrices.
\end{definition}

We use the game Jacobian to define a notion of stability of a Nash equilibrium $\NE$. It depends on the eigenvalues of the matrices $\Hs^{-1}\Jacobian(\NE)$, where $\Hs$ ranges over the family of positive-definite, block-diagonal matrices. If such an eigenvalue has positive real part, the dynamics around $\NE$ become unstable for some regularizers (\Cref{prop:non-convergence-general}). Due to that fact that the block diagonals of the game Jacobian are zero matrices, sum of all eigenvalues is zero: $\mathrm{tr}(\Hs^{-1}\Jacobian(\strats)) = 0$. As a result, instability also arises if there is an eigenvalue with negative real part---some other eigenvalue must have positive real part. And so, the only possibility for stability is when all eigenvalues of matrices $\Hs^{-1} \Jacobian(\NE)$ are purely imaginary. We call this uniform stability. Formally, this definition uses the following, purely linear-algebraic matrix condition (\Cref{def:uniform-stability-matrix}). Notice that uniform stability is a bilinear (that is, a second-order) condition depending only on the strategic component of the game.

\begin{definition}[Uniformly-stable block matrix]
    \label{def:uniform-stability-matrix}
    Let $\Jacobian \in \RR^{d \times d}$ be a block-matrix, where $d$ is the sum of its block dimensions $(d_1,\ldots, d_N)$. We say that $\Jacobian$ is \emph{uniformly stable} if the eigenvalues of $\Hs^{-1} \Jacobian$ are purely imaginary for all positive-definite, block-diagonal $\Hs \in \RR^{d \times d}$,
    \[\mathrm{spec}\big(\Hs^{-1}\Jacobian\big) \, \subseteq \, i\RR,\]
    where $\Hs = \mathrm{diag}(H_1,\ldots, H_N)$ and $H_n \in \RR^{d_n \times d_n}$ is positive definite.
\end{definition}

\begin{definition}[Uniformly stable equilibria] \label{def:uniform-stability}
    Let $(\domains, \utilities)$ be an $N$-player multilinear game. We say that a Nash equilibrium $\NE \in \domains$ is (pointwise) \emph{uniformly stable} if the Jacobian $\Jacobian(\NE)$ is uniformly stable. It is \emph{locally uniformly stable} if is contained in an open set $U \subset \domains$ such that $\Jacobian(\strats)$ is uniformly stable for all $\strats \in U$.
\end{definition}

\Cref{sec:economic-meaning} develops the economic meaning of uniform stability. \Cref{sec:convergence} shows the consequences of this spectral condition for dynamic stability. These two sections are independent of each other and can be read in either order. Before doing so, we provide some intuitive and formal motivation for uniform stability.

\subsection{Motivation for uniform stability}

The concept of uniform stability arises when we relax the assumption that players are `perfectly rational'. Unlike gradient-ascent players, smoothed best-response players do not necessarily update their strategies in the direction of steepest ascent (\Cref{lem:gradient-sbr}). With the existence of the regularizers, players are no longer maximally efficient---due to either bounded rationality or a need for better stability in their strategies, players only aim to improve modestly. 

At a current joint strategy $\strats$, let's say that the $n$th player makes an update along the $v_n$ direction. This leads to a local improvement whenever $\nabla_n f_n(\strats)^\top v_n > 0$. And this occurs if and only if there is a positive-definite matrix $H_n \succ 0$ so that the update direction $v_n$ and the gradient $\nabla_n f_n(\strats)$ are related by:
\[v_n = H_n^{-1} \nabla_n f_n(\strats).\]
This is a consequence of a simple linear-algebraic result, proved in \Cref{sec:pd-stretch}:

\begin{restatable}{lemma}{pdstretch} \label{lem:pd-stretch}
    Let $u, v \in \RR^m \setminus \{0\}$. Then, $u^\top v > 0$ if and only if there is a positive-definite matrix $H$ so that:
    \[u = Hv.\]
\end{restatable}

In optimization, this result justifies methods such as preconditioned gradient ascent and mirror ascent. In learning in games, it also directly justifies smoothed best-response dynamics as well as \emph{positive-definite adaptive dynamics} \citep{hopkins1999note}. In the current paper, this result leads to the notion of uniform stability by pointing us toward the preconditioned Jacobian maps $\strats \mapsto \big(H_1^{-1} \nabla_1f_1(\strats), \ldots, H_N^{-1}\nabla_N f_N(\strats)\big)$. As uncoupled players do not exactly know how others are learning, stability is defined uniformly over all positive-definite matrices. 
% This is a fairly strong condition; at a high level, this is the price that comes with considering a broad class of learning dynamics, where players may have bounded rationality.

A formal connection between uniform stability and smoothed best-response dynamics is obtained by computing the Jacobian of the smoothed best-response map $\betaBRs$. From \Cref{lem:gradient-sbr} below, it becomes clear that the eigenvalues of $\nabla \betaBRs(\NE)$ are related to the spectra of matrices $\Hs^{-1} \Jacobian(\NE)$. 
This result allows us to achieve in \Cref{sec:convergence} the convergence guarantees assuming local uniform stability, as well as non-convergence without pointwise uniform stability.
The proof of \Cref{lem:gradient-sbr} is given in \Cref{sec:gradient-sbf}.

\begin{restatable}[Gradient of smoothed best-response]{lemma}{gradientsbr}\label{lem:gradient-sbr}
    Let $(\domains, \utilities)$ be a multilinear game, where $\Omega_n \subset \RR^{d_n}$ for each $n \in [N]$, and let $\regs$ be a set of smooth and strictly convex regularizers, $h_n : \Omega_n \to \RR$. Then, the associated $\beta$-smoothed best-response map $\betaBRs$ is well-defined and smooth, where the gradient of $\betaBRs$ at $\strats$ is given by:
    \[\hspace{3em} \nabla \betaBRs(\strats) = \frac{1}{\beta} \Hs(\strats)^{-1} \Jacobian(\strats), \qquad \textrm{ where } \Hs(\strats)_{nm} = \begin{cases} \nabla_n^2 h_n(x_n) & n = m\\
    \hspace{7pt} 0_{d_n \times d_m} & n \ne m,\end{cases}\]
    and $\Hs(\strats) \in \RR^{d \times d}$ is a positive-definite, block-diagonal matrix with $d = d_1 + \dotsm + d_N$. Moreover, the gradient $\nabla \betaBRs$ is smooth in the interior of $\domains$.
\end{restatable}

Conversely, \Cref{lem:choice-regularizer} shows that for any choice $\strats$ and positive-definite, block diagonal matrix $\Hs^{-1}$, there exists a collection of steep regularizers such that $\nabla \betaBRs(\strats)$ is precisely given by $\Hs^{-1}\Jacobian(\strats)$, up to scaling.

%% file: doc/4-economic-meaning.tex
The previous section introduced uniform stability as a spectral notion---it is a property in terms of the eigenvalues of a family of matrices related to the game Jacobian. The main result of this paper, \Cref{thm:meaning-local}, shows that this spectral property is closely tied with the economic property of the equilibrium: if a mixed Nash equilibrium is locally uniformly stable, then it is locally strategically Pareto optimal.

To describe the proof strategy, let's consider a simple $N$-player multilinear game, where each player has one degree of freedom. In particular, let $f_n : \mathbb{R}^N \to \mathbb{R}$ be the utility for player $n \in [N]$, and let the origin $\mathbf{0} \in \mathbb{R}^N$ be a Nash equilibrium. Let's assume that $f_n$ is purely strategic. In this case, it has the form
\[\hspace{5em}f_n(\mathbf{x}) = x_n g_n(x),\qquad g_n(\mathbf{x}) = \partial_n f_n(\mathbf{x}).\]
At the Nash equilibrium, we have that $f_n(\mathbf{0}) = 0$ and  $g_n(\mathbf{0}) = 0$.

The statement of \Cref{thm:meaning-local} for this game would say that if the origin is locally uniformly stable, then it is locally weakly Pareto optimal. That is, there is an open set $U \subset \mathbb{R}^N$ containing the origin such that for all alternative strategies $\strats \in U$, there is a player $k \equiv k(\mathbf{x})$ whose utility does not improve
\[f_k(\strats) = x_k g_k(\mathbf{x}) \leq 0.\]
We show this by proving that the function $-G$ where $G = (g_1,\ldots, g_N)$ is a \emph{$\mathrm{P}_0$-function}. This can be considered as a type of generalization of monotonicity:

\begin{definition}[$\mathrm{P}_0$-function, \cite{more1973pfunctions}]
    A function $Q \equiv (q_1,\ldots, q_N): D \subset \mathbb{R}^N \to \mathbb{R}^N$ is a \emph{$\mathrm{P}_0$-function} if for all distinct $\mathbf{x},\mathbf{y} \in D$, there is an index $k = k(\mathbf{x},\mathbf{y}) \in [N]$ such that
    \[(x_k - y_k)[q_k(\mathbf{x}) - q_k(\mathbf{y})] \geq 0.\]
\end{definition}

Note that if, in this definition, we let $\mathbf{y}$ be the Nash equilibrium, the condition that $-G$ is a $\mathrm{P}_0$-function precisely shows that it the equilibrium is weakly Pareto optimal. In order to show that $-G$ is $\mathrm{P}_0$, we will first show that the negative game Jacobian $-\mathbf{J}$ is a $\mathrm{P}_0$-matrix whenever it is uniformly stable. A \emph{$\mathrm{P}_0$-matrix} is a generalization of positive semi-definiteness to asymmetric square matrices:

\begin{definition}[$\mathrm{P}_0$-matrix, \cite{fiedler1966some}]
    A matrix $A \in \mathbb{R}^{N \times N}$ is a \emph{$\mathrm{P}_0$-matrix} if all principal minors of $A$ are non-negative. That is, for all subsets $S \subset [N]$, the determinant of the principal submatrix $A[S^c]$ is non-negative, where $A[S^c]$ is the matrix $A$ after removing all rows and columns with index in $S$.
\end{definition}

A classical result shows that a smooth function is $\mathrm{P}_0$ if its Jacobian is $\mathrm{P}_0$ everywhere:

\begin{theorem}[Corollary 5.3 of \cite{more1973pfunctions}] \label{thm:P0-function}
    Let $Q : [0,1]^N \to \mathbb{R}^N$ be Fréchet differentiable. If the Jacobian $Q'(\mathbf{x})$ is a $\mathrm{P}_0$-matrix for each $\mathbf{x} \in [0,1]^N$, then $Q$ is a $\mathrm{P}_0$-function.
\end{theorem}

% \begin{theorem}[\cite{fiedler1966some}]
%     A matrix $A \in \mathbb{R}^{N \times N}$ is $\mathrm{P}_0$ if and only if for each vector $x \ne 0$, there is an index $k \in [N]$ such that $x_k \ne 0$ and $x_k y_k \geq 0$ where $y = Ax$.
% \end{theorem}

It is not too difficult to see why uniformly stable real matrices must be $\mathrm{P}_0$. Uniform stability implies that all eigenvalues of any principal submatrix is also purely imaginary. Then, as each non-zero imaginary eigenvalue $i \lambda$ comes with a conjugate pairing $- i \lambda$, their product is $-i^2\lambda^2 = \lambda^2 \geq 0$ is non-negative. The corresponding principal minor is the product of all of these eigenvalues, so it is also non-negative. 

Finally, to see why all eigenvalues of any principal submatrix are imaginary, fix any subset $S \subset [N]$. Let $D_S^\varepsilon$ be the diagonal matrix $\mathrm{diag}(d_1^\varepsilon,\ldots, d_N^\varepsilon)$ where
\[d_i^\varepsilon = \begin{cases}
    \varepsilon & i \in S \\
    1 & i \notin S.
\end{cases}\]
By uniform stability, the matrix product $D_S^\varepsilon A$ has purely imaginary eigenvalues for all $\varepsilon > 0$. The limit of $D_S^0 A$ must also have a purely imaginary spectrum, by the continuity of eigenvalues. This matrix $D_S^0 A$ can also be constructed by taking $A$ and zeroing out all the rows and columns with indices in $S$. It follows that the eigenvalues of the principal submatrix $A[S^c]$ are contained in $\mathrm{spec}(D_S^0A)$ and thus are purely imaginary.

%% file: doc/5-convergence-mixed.tex
The fundamental result by \cite{hart2003uncoupled} shows that, in general, uncoupled dynamics cannot be guaranteed to asymptotically converge to mixed Nash equilibria. In this section, we refine this result for smoothed best-response learning dynamics \eqref{eqn:dynamics}, providing conditions for when the dynamics are unstable and fail to converge, and when they are stable and converge. We consider two stability classes for mixed equilibria: those that are pointwise uniformly stable, and those that are locally uniformly stable.

For simplicity, we consider a Nash equilibrium $\NE$ that is approximated by the sequence of $\beta$-smoothed equilibria $\SE$ as $\beta$ goes to zero (see \Cref{lem:convergence-to-Nash}). The first result shows that when $\NE$ is \emph{not} uniformly stable, then the $(\betaBRs,\eta)$-averaging dynamics can become unstable around $\SE$ once $\beta$ becomes sufficiently small. Intuitively, this means that $\NE$ is `inapproximable' by generic smoothed best-response dynamics. In contrast, our second result shows that if $\NE$ is contained in a region that is uniformly stable, then it can be approximated to arbitrary precision by generic, smoothed best-response dynamics. It turns out that the price of higher precision, and in turn smaller $\beta$, is a slower rate of convergence due to the requirement of a smaller learning rate $\eta$. Informally:

\begin{center}
\begin{minipage}{0.95\linewidth}
    \textbf{Non-convergence result.} If a Nash equilibrium $\NE$ is not uniformly stable, then there are smoothed best-responses that cannot be stabilized: there are regularizers $\regs$ so that when $\beta > 0$ is sufficiently small, all $(\betaBRs, \eta)$-averaging dynamics become unstable about the smoothed equilibrium $\SE$ (\Cref{prop:non-convergence-general}).
\end{minipage}
\end{center}
\begin{center}
\begin{minipage}{0.95\linewidth}
    \textbf{Convergence result.} If a Nash equilibrium is $\NE$ is contained in a region that is uniformly stable, then every smoothed best-responses can be stabilized: for all $\betaBRs$ and sufficiently small $\eta > 0$, the $(\betaBRs, \eta)$-averaging dynamics locally converges to $\SE$ (\Cref{thm:main-convergence}).
\end{minipage}
\end{center}

\begin{figure}[t]
    \centering
    \begin{minipage}{0.44\linewidth}
    \centering
    
    \includegraphics[width=0.95\linewidth]{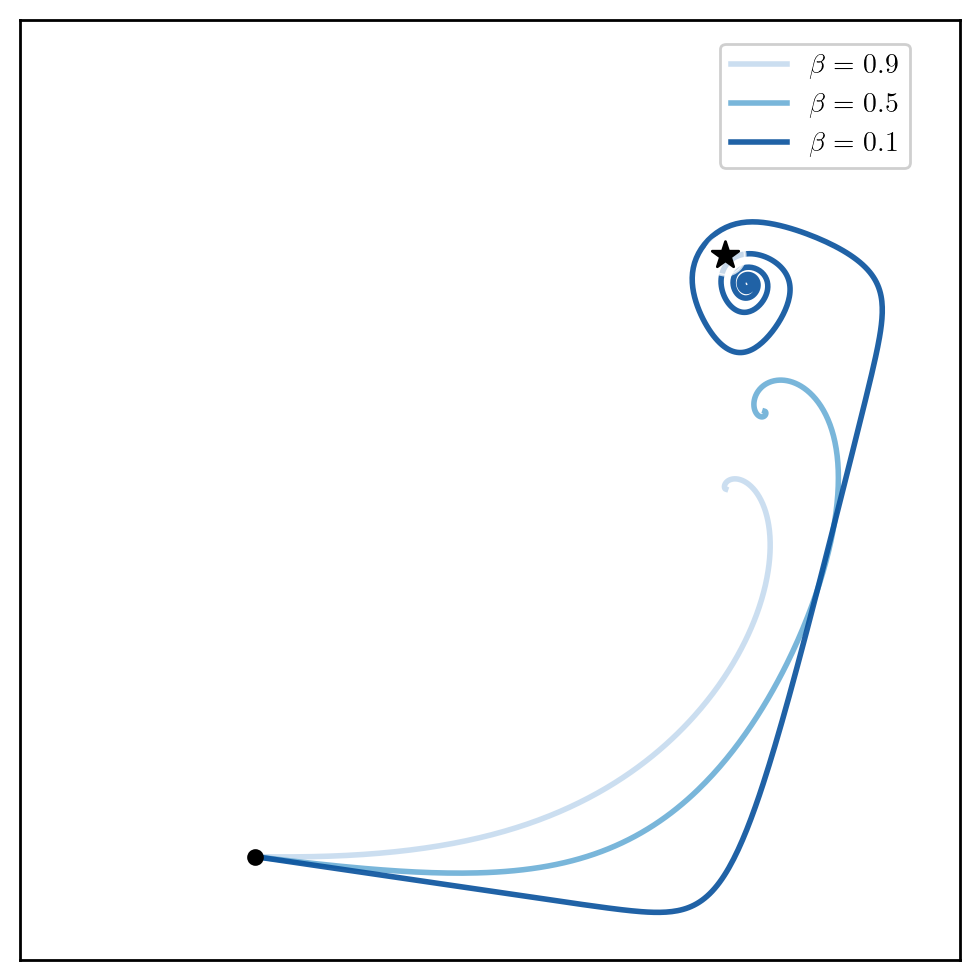}

    \small (a) Dynamics across smoothing parameters $\beta$
    \end{minipage}\hspace{10pt}
    \begin{minipage}{0.44\linewidth} 
    \centering
    
    \includegraphics[width=0.95\linewidth]{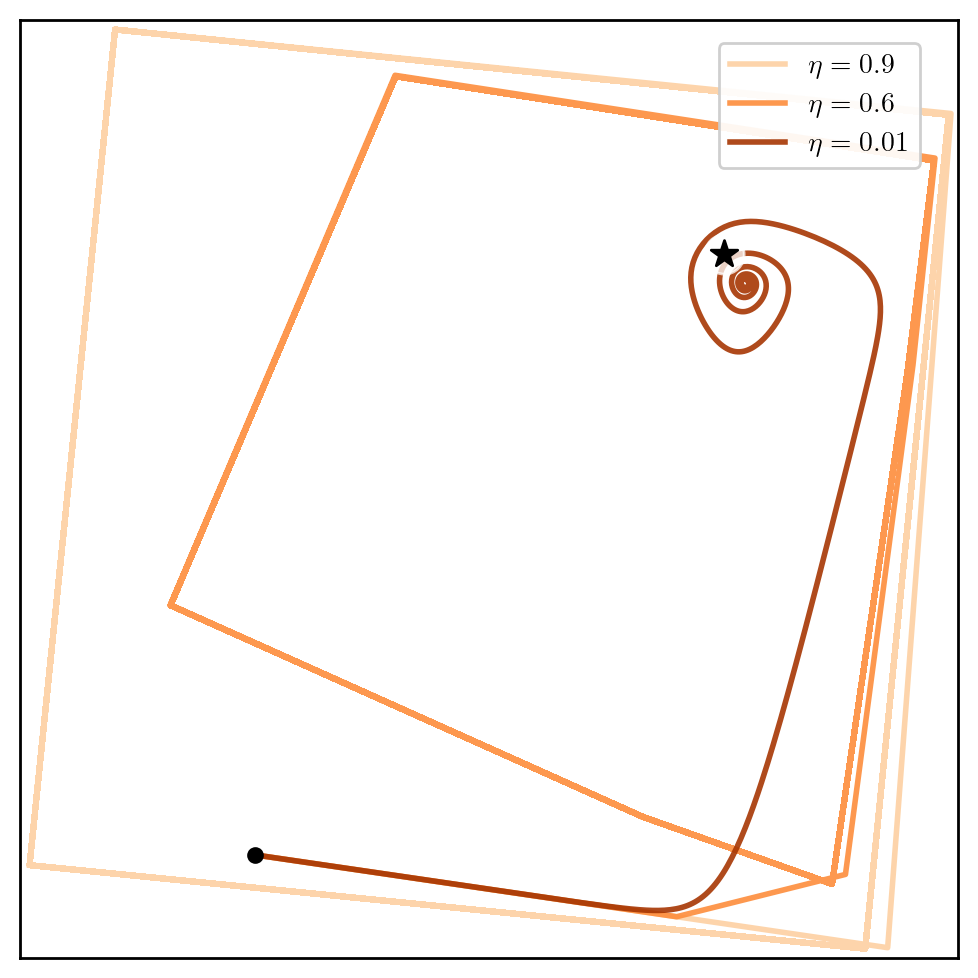}

    \small  (b) Dynamics across learning rates $\eta$
    \end{minipage}
    
    \caption{The trajectories of $\beta$-smoothed best-response dynamics with $\eta$-learning rate toward a uniformly-stable, mixed Nash equilibrium (star) in a two-player normal-form game. (a) The $\beta$-smoothed equilibria become better approximations of the Nash equilibrium as $\beta$ shrinks, but the dynamics become less stable and exhibit more cycling. The figure on the left plots the trajectories initialized at the black dot for varying smoothing $\beta$ and fixed averaging $\eta$. (b) The learning dynamics converge to $\beta$-smoothed equilibria once $\eta$ becomes sufficiently small, and it does so with greater stability with smaller $\eta$. However, stability comes at the expense of slower rate of convergence. The figure on the right plots the trajectories for fixed $\beta$ and varying $\eta$.}
    \label{fig:eta-beta-visualization}
\end{figure}

\subsection{Stability concepts in dynamical systems}
Before we state the main results, recall the standard notion of stability from dynamical systems: 

\begin{definition}[Stability of fixed points] \label{def:stability}
    Let $\domains \subset \RR^d$ be a domain and let $\BR : \domains \to \domains$ be a transition map defining the discrete-time dynamical system $\strats(t+1) = \BR(\strats(t))$. Let $\NE \in \domains$ be a fixed point of $\BR$.
    \begin{itemize}
        \item The fixed point $\NE$ is \emph{stable} if for every $\epsilon > 0$, there exists some $\delta > 0$ such that:
        \[\phantom{\forall t \geq 0}\big\|\strats(0) - \NE\big\| < \delta \qquad \implies \qquad \forall t \geq 0, \quad \big\|\strats(t) - \NE\big\| < \epsilon.\]
        \item Moreover, it is \emph{asymptotically stable} if there exists some $\delta > 0$ such that:
        \[\phantom{\forall t \geq 0}\big\|\strats(0) - \NE\big\| < \delta \qquad \implies \qquad \lim_{t \to \infty} \,\big\|\strats(t) - \NE\big\| = 0.\]
    \end{itemize}
    We say that $\NE$ is \emph{unstable} if it is not stable.
\end{definition}

We will consider families of $\beta$-smoothed best-response dynamics, whose fixed points $\SE$ converge to a Nash equilibrium $\NE$ as $\beta$ goes to zero. This motivates the following notion of \emph{stabilization} to Nash equilibria, where each $\SE$ is an asymptotically-stable fixed point of sufficiently fine $(\betaBRs, \eta)$-averaging dynamics.

\begin{definition}[Stabilizing smoothed best-response dynamics]
    Let $(\domains, \utilities)$ be an $N$-player multilinear game with a Nash equilibrium $\NE$. Let $\regs$ be a set of smooth and strictly convex regularizers. We say that smoothed best-response dynamics with respect to $\regs$ \emph{stabilizes} to $\NE$ if the following hold:
    \begin{enumerate}
        \item[(a)] The $\beta$-smoothed equilibria $\SE$ converge to $\NE$ as $\beta$ goes to zero.
        \item[(b)] For each $\beta > 0$, the $\beta$-smoothed equilibrium $\SE$ is an asymptotically-stable fixed point of the $(\betaBRs, \eta)$-averaging dynamics \eqref{eqn:dynamics} whenever $\eta \in (0,1)$ is sufficiently small.
    \end{enumerate}
\end{definition}

\subsection{Convergence results for smoothed best-response dynamics} \label{subsec:non-convergence}

The first result shows that if a Nash equilibrium is not uniformly stable, then not only is it not possible to stabilize generic smooth best-response dynamics to it (the dynamics to $\SE$ are not asymptotically stable for arbitrarily small $\beta$), but in fact, the dynamics can become unstable. The following is proved in \Cref{sec:proof-non-convergence}.

\begin{restatable}[Inapproximability of unstable equilibria]{proposition}{inapprox} \label{prop:non-convergence-general}
    Let $(\domains, \utilities)$ be an $N$-player multilinear game with an interior Nash equilibrium $\NE$. Suppose that $\NE$ is not uniformly stable. Then, there is some collection of smooth and strictly convex regularizers $\regs$ such that smoothed best-response dynamics does not stabilize to $\NE$. In fact, even if the sequence of $\beta$-smoothed equilibria $\SE$ converges to $\NE$, each $\SE$ is an unstable fixed point of the $(\betaBRs, \eta)$-averaging dynamics \eqref{eqn:dynamics} for all $\eta \in (0,1)$ and sufficiently small $\beta > 0$.
\end{restatable}

Now, we present the main convergence result concerning Nash equilibria with local uniform stability. The following is proved in \Cref{sec:proof-main-convergence}.

\begin{restatable}[Convergence to equilibria in uniformly-stable regions]{theorem}{mainconv} \label{thm:main-convergence}
    Let $(\domains,\utilities)$ be an $N$-player multilinear game with a unique Nash equilibrium $\NE$ that is locally uniformly stable. Over arbitrary choice of smooth and strictly convex regularizer $\regs$, all smoothed best-response dynamics locally stabilize to $\NE$. There is a constant $W > 0$ such that the following holds. For each choice of $\beta > 0$ such that $\SE$ is uniformly stable, there exists a radius $r_\beta > 0$ and constant $C_\beta > 0$ such that whenever $\|\strats(0) - \SE\| < r_\beta$, the $(\betaBRs,\eta)$-averaging dynamics converge when the learning rate is bounded $\smash{\eta \leq \frac{1}{2} \frac{\beta^2}{\beta^2 + W^2}}$:
    \[\|\strats(0) - \SE\| \leq C_\beta \exp\left( - \frac{\eta t}{2}\right).\]
    In particular, for sufficiently small $\beta$, the right-hand side can be made less than any $\varepsilon > 0$ whenever
    \[t = \Omega\left(\frac{1}{\beta^2} \bigg( d \log d + d \log \frac{1}{\beta}+ \log \frac{1}{\epsilon}\bigg)\right),\]
    where $d = d_1 + \dotsm + d_N$ is the degree of freedom across all players.
\end{restatable}

\begin{comment}
\begin{restatable}[Convergence to equilibria in uniformly-stable regions]{theorem}{mainconv} \label{thm:main-convergence}
    Let $(\domains, \utilities)$ be an $N$-player multilinear game with a Nash equilibrium $\NE$ that is locally uniformly stable. Over arbitrary choice of smooth and strictly convex regularizer $\regs$, all smoothed best-response dynamics stabilize to $\NE$. In particular, there is a constant $C_{\utilities} > 0$ depending only on the game so that when $\eta \leq C_{\utilities} \beta^2$, the $(\betaBRs, \eta)$-averaging dynamics globally converge: 
    \[\big\|\strats(t) - \SE\big\| \leq \exp\left( - \frac{\eta t - \ln N}{2}\right).\]
\end{restatable}
\end{comment}

Using the fact that $\beta$-smoothed equilibria are $O(\beta)$-approximate Nash equilibria, this implies that smoothed best-response dynamics with averaging achieve $T^{-1/2}$-convergence rate to Nash equilibria.

%% file: doc/6-convergence-partial.tex
The main convergence result (\Cref{thm:main-convergence}) can be extended to equilibria that are not fully mixed, and are not interior equilibria. A Nash equilibrium $\NE$ is on the boundary of the probability simplex if there are players who never play certain actions. This can happen if these actions are strictly dominated, in which case a perfectly rational player will never play those strategies, while a player with bounded rationality will play them with decreasing frequency. Intuitively, as long as the players are not overly irrational, the stability of smoothed best-response dynamics around these equilibria should not hinge upon these suboptimal actions. 

Before we make this intuition formal, recall some notations for normal-form game:
\begin{itemize} 
    \item The joint strategy space $\domains = \Omega_1 \times \dotsm \times \Omega_N$ is a product of probability simplices, $\Omega_n = \Delta^{k_n - 1}$, where the $n$th player may randomize over $k_n$ pure strategies.
    \item The vertices of the simplex $\Omega_n$ correspond to pure strategies of player $n$. For each alternative $i \in [k_n]$, we let $e_{n,i} \in \Omega_n$ denote the pure strategy where the $n$th player deterministically chooses the $i$th action.
\end{itemize}
The \emph{support} of a strategy is the set of pure strategies with positive probability of being played:

\begin{definition}[Support of a strategy]
    For any player $n$, the \emph{support} of a strategy $x_n \in \Omega_n$ is the set:
    \[\mathrm{supp}(x_n) = \big\{i \in [k_n] : x_{n,i} \ne 0 \big\}.\]
    Given $\strats \in \domains$, define the reduced strategy sets:
    \begin{equation}  \label{eqn:strategy-restriction}
        \Omega_n(x_n) := \big\{x_n' \in \Omega_n : \mathrm{supp}(x_n')\subset \mathrm{supp}(x_n)\big\} \qquad \textrm{and}\qquad \domains(\strats) := \prod_{i \in [N]} \Omega_n(x_n).
    \end{equation} 
    Note that $\Omega_n(x_n)$ is an $(k_n' - 1)$-dimensional simplex where $k_n' = |\mathrm{supp}(x_n)|$.
\end{definition}

For the convergence result, we consider \emph{quasi-strict equilibria}, which are Nash equilibria where each player fully mixes on the set of best responses \citep{harsanyi1988general,van1992refinements}.\footnote{Classically, these are called \emph{quasi-strong} equilibria. We prefer the term \emph{quasi-strict} since \emph{strong Nash equilibria} is another, unrelated refinement of the Nash equilibrium concept, which we also make use of. For discussion, see \cite{van1992refinements}.} This captures the property that if an alternative is never played in an equilibrium, then it must be strictly dominated.

\begin{definition}[Quasi-strict equilibrium]
    We say that a Nash equilibrium $\NE$ of an $N$-player normal form game is \emph{fully-mixed on the best-response set} or \emph{quasi-strict} if for each player $n \in [N]$:
    \[i \in \mathrm{supp}(x_n^*) \qquad \Longleftrightarrow \qquad i \in \argmax_{i \in [k_n]}\, f_n(e_{n,i}; \NE_{-n}) \]
\end{definition}

Since actions that are not supported by a quasi-strong equilibrium are dominated, we expect that they are strategically irrelevant. This motivates the following refinement of the game \citep{harsanyi1988general}, where the irrelevant actions are removed from the game. 

\begin{definition}[Reduced game]
    Let $(\domains, \utilities)$ be a normal-form game with a Nash equilibrium $\NE$. The \emph{reduced game} around $\NE$ is the game on the domain $\domains(\NE)$ with the utilities $\utilities$ restricted to $\domains(\NE)$.
\end{definition}

In the reduced game, $\NE$ becomes an interior or fully-mixed Nash equilibrium (\Cref{cor:reduction-eq}), allowing the results in \Cref{subsec:non-convergence} to be applied. The remaining problem is about how the players learn to converge towards the support of the equilibrium and avoid playing the suboptimal strategies that have been removed from the reduced game. In order to control the rate at which these suboptimal strategies appear in the $\beta$-smoothed equilibria $\SE$, we introduce the following \emph{quantitative} version of steepness (c.f.\ \Cref{def:steep-reg}):

\begin{definition}[Linear steepness]
    \label{def:linear-steepness}
    Let $h : \Delta^{k-1} \to \RR$ be a steep regularizer. For all $v \in \RR^k$, define the map: 
    \begin{equation*}
        x^\beta(v) = \argmax_{x \in \Delta^{k-1}}\, v^\top x - \beta h(x).
    \end{equation*}
    For each $i \in [i]$ and $\epsilon > 0$, let $V_i(\epsilon) = \{v \in \RR^k : v_i < \max_{j \in [k]} v_j - \epsilon\}$ be the set of vectors where the $i$th component is $\epsilon$-suboptimal.  We say that $h$ is \emph{linearly steep} if the following holds for all $i \in [k]$ and $\epsilon > 0$:
    \[\lim_{\beta \to 0}\, \sup_{v \in V_i(\epsilon)} \frac{x^\beta(v)_i}{\beta} =0.\]
\end{definition}

In words, this states that the probability mass $x^\beta(v)_i$ placed on any alternative $i \in [k]$ must shrink sublinearly with $\beta$ when the value $v_i$ is $\epsilon$-suboptimal. In the case that $h$ is the entropy map, \Cref{ex:entropy-linear-steep} shows that, in fact, this distance shrinks as $\exp(-\Omega(1/\beta))$ at a much faster, exponential rate. As a consequence of linear steepness, the following shows that the rate at which suboptimal alternatives are played in $\beta$-smoothed equilibria $\SE$ also vanish at a sublinear rate $o(\beta)$:

\begin{restatable}[Dominated strategies occur at $\beta$-sublinear rates in $\beta$-smoothed equilibria]{lemma}{dominatedrate}\label{lem:steep-convergence}
    Let $(\domains, \utilities)$ be an $N$-player normal-form game and $\regs$ be a collection of linearly steep regularizers. Suppose that as $\beta$ goes to zero, the sequence of of $\beta$-smoothed equilibria $\SE$ converges to a quasi-strict equilibrium $\NE$. Let $\Pi : \domains \to \domains(\NE)$ be the orthogonal projection onto the support of $\NE$. Then:
    \[\lim_{\beta \to 0}\, \frac{\big\|\SE - \Pi \SE\big\|}{\beta} = 0.\]
\end{restatable}

Before we can give the main convergence result for all quasi-strict equilibria, we need a technical condition to ensure that pseudoinverse $\nabla^2 h(x)^{+}$ extends smoothly to the boundary of the simplex:

\begin{definition}[Proper regularizer] \label{def:extended-smoothness}
    A map $h : \Delta^{k-1} \to \RR$ on the simplex is a \emph{proper regularizer} if:
    \begin{itemize}
        \item The restriction of $h$ onto each face of the simplex is steep regularizer.
        \item The Moore-Penrose pseudoinverse of the Hessian $\nabla^2 h(x)^{+}$ is smooth on the simplex.
    \end{itemize}
\end{definition}

See \Cref{sec:calculus} for formal definition of the \emph{faces} of the simplex and what it means to be \emph{smooth} on it. The following convergence result shows that if $\NE$ is a quasi-strict equilibrium, it suffices to analyze its uniform stability within the reduced game, neglecting all suboptimal alternatives. This is because whenever the regularizers are linearly steep, any potentially de-stabilizing effect of these additional strategies are overwhelmed by the stabilizing effects of the averaging dynamics.

\begin{restatable}[Convergence to non-interior equilibria]{theorem}{boundaryconv}\label{thm:boundary-convergence}
    Let $(\domains,\utilities)$ be an $N$-player, normal-form game and let $\regs$ be a collection of linearly steep, proper regularizers. Suppose that $\NE$ is a quasi-strict Nash equilibrium and that it is the limit of the $\beta$-smoothed equilibria $\SE$ as $\beta$ goes to zero. If the reduced game around $\NE$ is locally uniformly stable, then smoothed best-response dynamics stabilize to $\NE$. In particular, there exists a constant $C_\utilities > 0$ such that for all sufficiently small $\beta > 0$ and $\eta \leq C_\utilities \beta^2$,  the $(\betaBRs, \eta)$-averaging dynamics is locally asymptotically stable at $\SE$, where the operator norm of its Jacobian at $\SE$ is bounded by:
    \[\big\|(1 - \eta) \mathbf{I} + \eta \nabla \betaBRs(\SE)\big\|_2 \leq \exp\left(- \frac{\eta}{2}\right).\]
\end{restatable}

The results of this section are proved in \Cref{sec:beyond-interior}.

%% file: doc/7-discussion.tex
In this work, we introduce a theory of non-asymptotic stability for learning in games. We demonstrate via this theory a connection between uniform stability and collective rationality. On the learnability of the mixed Nash equilibria, we connect uniform stability with the last-iterate convergence behavior for the class of incremental, smoothed best-response dynamics. We close with a few open questions that we believe would be important to pursue.

\paragraph{Strengthening the necessity of uniform stability for convergence to mixed Nash equilibria} Our main non-convergence result (\Cref{prop:non-convergence-general}) shows that if an equilibrium $\NE$ is not uniformly stable, then there exist smoothed best-responses that cannot be stabilized. We believe this can be strengthened, where  `there exist' is replaced with `almost all' (in terms of the Hessian of the regularizer at $\NE$).

%%%%%%% Need to modify below (Perhaps change to ``Studying the economic meanings of deviation from unstable mixed Nash equilibria'': e.g., whether there is always a group of players that improve, or whether there is a relationship in the utilities among all the players.)
A second strengthening has to do with the economic properties of the non-converging behaviors. We have shown that players can avoid stabilizing to Pareto inefficient equilibria, which can be seen as a form of collective rationality. However, we have not analyzed if players escape inefficient equilibria in a collectively-rational way:

\vspace{7pt}
\noindent \textbf{Open Question 1.} {\itshape Are there uncoupled learning dynamics that not only do not stabilize to mixed equilibria that are not uniformly stable, but moreover locally escape them in collectively-rational ways. That is, whenever $\|\strats(t) - \NE\| < \epsilon$ is sufficiently close, there exists some $\delta > 0$ such that for some time, all players improve:
\[\forall s \in (t, t+ \delta) \quad \textrm{and}\quad n \in [N],\qquad f_n\big(\strats(s)\big) > f_n\big(\strats(t)\big).\]
}

%%%%%%%%%%%%%%%%%%%%%%%%%%%%%%%%%%%%%%%%%%%%%%%%%%%%

\paragraph{Non-asymptotic analyses beyond smoothed best-response dynamics} Smoothed best-response dynamics is a natural family of learning dynamics. It would be of interest to consider other classes of updates that can lead to mixed Nash equilibria. 

% \paragraph{Relaxation of bi-directionality assumption} It would be interesting to generalize beyond the assumption of bi-directional interactions. In particular, the bi-directional interactions are represented by undirected graphs. One can also consider interaction structures that are represented by directed graphs, which may lead to longer-range interactions (e.g., players form a directed cycle as in the Cyclic Matching Pennies game of \cite{kleinberg2011beyond}).